\documentclass{IEEEtran}
\usepackage[colorlinks,urlcolor=blue,linkcolor=blue,citecolor=blue]{hyperref}

\usepackage{color,array}

\usepackage{graphicx}
\graphicspath{{figs/}}
\usepackage{authblk}

\usepackage{fullpage}

\usepackage{mathtools,amsmath,amssymb,amsfonts,amsthm}

\usepackage{lipsum}
\usepackage{xcolor}
\usepackage{xspace}
\usepackage{amsmath}
\usepackage[capitalize]{cleveref}
\usepackage{thm-restate}
\usepackage{algorithm}
\usepackage[noend]{algpseudocode}

\usepackage{enumitem}
\usepackage{tikz}
\usepackage{caption}

\usepackage{mathrsfs}
\usepackage{csquotes}

\newcommand{\scrC}{\mathscr{C}}
\newcommand{\cS}{\mathcal{S}}
\newcommand{\cT}{\mathcal{T}}

\newcommand{\N}{\mathbb{N}}
\newcommand{\R}{\mathbb{R}}

\newcommand{\Dc}{\mathcal{D}}

\newcommand{\Mc}{\mathcal{M}}

\newcommand{\Sc}{\mathcal{S}}
\newcommand{\Tc}{\mathcal{T}}
\newcommand{\Cc}{\mathcal{C}}

\newcommand{\Wc}{\mathcal{W}}
\newcommand{\Xc}{\mathcal{X}}

\definecolor{myred}{rgb}{0.8,0.2,0.2}
\definecolor{myblue}{rgb}{0.2,0.7,0.2}
\definecolor{mygreen}{rgb}{0.2,0.2,0.7}

\DeclarePairedDelimiterX{\inp}[2]{\langle}{\rangle}{#1, #2}

\newcommand{\AP}{\ensuremath{\mathtt{AP}}}
\newcommand{\cRWA}{\ensuremath{\mathtt{cRWA}}}
\newcommand{\RWA}{\ensuremath{\ensuremath{\Omega}}}
\newcommand{\RWAa}{\ensuremath{\ensuremath{\RWA_a}}}
\newcommand{\RWAe}{\ensuremath{\ensuremath{\RWA_e}}}
\newcommand{\contextRWA}{\ensuremath{\kappa}}
\newcommand{\reachRWA}{\ensuremath{\mathcal{R}}}
\newcommand{\avoidRWA}{\ensuremath{\mathcal{A}}}
\newcommand{\avoidRWAa}{\ensuremath{\avoidRWA_a}}
\newcommand{\avoidRWAe}{\ensuremath{\avoidRWA_e}}

\newcommand{\allU}{\ensuremath{\mathfrak{U}}}
\newcommand{\allCLF}{\ensuremath{\mathfrak{W}}}
\newcommand{\allRWA}{\ensuremath{\cRWA(\game^I,\assump^I)}}

\theoremstyle{plain}
\newtheorem{theorem}{Theorem}
\newtheorem{corollary}{Corollary}
\newtheorem{proposition}{Proposition}
\newtheorem{lemma}{Lemma}

\theoremstyle{definition}
\newtheorem{definition}{Definition}
\newtheorem{assumption}{Assumption}
\newtheorem{problem}{Problem}
\newtheorem{example}{Example}
\newtheorem{remark}{Remark}

\makeatletter 
\newcommand{\AllQ}{\@ifstar\AllQStar\AllQNoStar}
\newcommand{\AllQStar}[3][\;]{\ensuremath{\left(\forall #2#1.#1#3\right)}}
\newcommand{\AllQNoStar}[3][\;]{\ensuremath{\forall #2#1.#1#3}}
\newcommand{\AllQu}{\@ifstar\AllQuStar\AllQuNoStar}
\newcommand{\AllQuStar}[3][\;]{\ensuremath{\left(\forall^{\infty} #2#1.#1#3\right)}}
\newcommand{\AllQuNoStar}[3][\;]{\ensuremath{\forall^{\infty} #2#1.#1#3}}

\newcommand{\ExQ}{\@ifstar\ExQStar\ExQNoStar}
\newcommand{\ExQStar}[3][\;]{\ensuremath{\left(\exists #2#1.#1#3\right)}}
\newcommand{\ExQNoStar}[3][\;]{\ensuremath{\exists #2#1.#1#3}}
\makeatother

\newcommand{\ldist}{\Upsilon}
\newcommand{\Ldist}{\ensuremath{\mathfrak{D}}}
\newcommand{\trace}{\pi}
\newcommand{\traces}{\mathsf{Traces}}
\newcommand{\labelfunc}{L}
\newcommand{\labelfuncG}{\ell}
\newcommand{\true}{\texttt{true}\xspace}
\newcommand{\false}{\texttt{false}\xspace}
\newcommand{\until}{\,\mathtt{U}\,}
\newcommand{\weakuntil}{\,\mathtt{W}\,}
\newcommand{\Next}{\xspace\bigcirc\xspace}
\newcommand{\finally}{\xspace\Diamond\xspace}
\newcommand{\globally}{\xspace\square\xspace}

\newcommand{\vertices}{V}
\newcommand{\vertex}{v}
\newcommand{\altvertex}{v'}
\newcommand{\verticeso}{V_0}
\newcommand{\verticesl}{V_1}
\newcommand{\verticesi}{V_i}
\newcommand{\edges}{E}
\newcommand{\edgeso}{E_0}

\newcommand{\edgesi}{E_i}
\newcommand{\cupdot}{\mathbin{\mathaccent\cdot\cup}}

\newcommand{\abs}[1]{\left\lvert#1\right\rvert}
\newcommand{\game}{\mathcal{G}}
\newcommand{\gamegraph}{\ensuremath{G}}
\newcommand{\spec}{\phi}
\newcommand{\specF}{\spec}
\newcommand{\specS}{\varphi}
\newcommand{\play}{\rho}
\newcommand{\p}[1]{\ensuremath{\text{Player}~#1}}

\newcommand{\priority}{\ensuremath{\mathbb{P}}}
\newcommand{\paritygame}{\textsc{Parity}}
\newcommand{\src}{\textsc{src}}

\newcommand{\attr}[3]{\textsf{attr}\ensuremath{^{#3}\left(#1,#2\right)}}
\newcommand{\attri}[2]{\attr{#1}{#2}{i}}
\newcommand{\attrl}[2]{\attr{#1}{#2}{1}}
\newcommand{\attro}[2]{\attr{#1}{#2}{0}}

\newcommand{\Attr}[3]{\textsc{Attr}\ensuremath{^{#3}\left(#1,#2\right)}}
\newcommand{\Attri}[2]{\Attr{#1}{#2}{i}}
\newcommand{\Attrl}[2]{\Attr{#1}{#2}{1}}
\newcommand{\Attro}[2]{\Attr{#1}{#2}{0}}

\newcommand{\pre}{\textsf{pre}}
\newcommand{\win}{\vertices_{\text{win}}}

\newcommand{\strat}{\sigma}

\newcommand{\safeset}{\mathbb{S}}
\newcommand{\ellipsoid}{\mathbb{E}}
\newcommand{\polyhedron}{\mathbb{H}}
\newcommand{\tr}{\mathrm{tr}}
\newcommand{\cols}{\mathrm{cols}}

\newcommand{\livegroup}{\mathcal{H}}

\newcommand{\livegroupSingleN}{H}
\newcommand{\colivegroup}{D}
\newcommand{\safegroup}{S}
\newcommand{\assump}{\ensuremath{\psi}}
\newcommand{\assumpsafe}{\ensuremath{\assump_{\textsc{unsafe}}}}

\newcommand{\assumpgrlive}{\ensuremath{\assump_{\textsc{live}}}}
\newcommand{\assumpdep}{\ensuremath{\assump_{\textsc{colive}}}}

\newcommand{\perslivegroups}{\Lambda}
\newcommand{\persedges}{\ensuremath{\mathtt C}}
\newcommand{\perssource}{\ensuremath{\mathtt S}}
\newcommand{\perstarget}{\ensuremath{\mathtt T}}
\newcommand{\assumpC}{\ensuremath{\assump_{\textsc{cont}}}}
\newcommand{\assumpPers}{\ensuremath{\assump_{\textsc{pers}}}}
\newcommand{\reach}{\textsc{SolveReach}}
\newcommand{\solve}{\textsc{Solve}}

\newcommand{\sol}{\xi}

\colorlet{darkgreen}{green!80!black}
\colorlet{darkred}{red!80!black}
\usetikzlibrary{arrows, automata, shapes,positioning,calc}
\tikzset{auto, >= stealth}
\tikzset{every edge/.append style={thick, shorten >= 1pt}}
\tikzset{initial/.style={draw, thick, <-, shorten <=1pt}}
\tikzset{player0/.style = {draw, thick, shape=circle, minimum size=5mm}}
\tikzset{player1/.style = {draw, thick, shape=rectangle, minimum size=5mm}}
\tikzset{bplayer0/.style = {draw, thick, shape=ellipse, minimum size=3mm,text width=0.55cm}}
\tikzset{bplayer1/.style = {draw, thick, shape=rectangle, minimum size=3mm,text width=0.8cm}}

\newcommand\bhpos{2}
\newcommand\bypos{1.2}
\newcommand\hpos{1.5}
\newcommand\ypos{1.6}
\newcommand*\bcircled[1]{~\tikz[baseline=(char.base)]{\node[shape=circle,draw, text width=3mm,align=center, minimum size = 2mm, inner sep = 0pt, fill = black] (char) {\scriptsize \textcolor{white}{\textbf{#1}}}}}

\tikzset{box/.style={rectangle, draw, minimum height=1cm, fill=black!30}}
\tikzset{boxl/.style={rectangle, draw, minimum height=1cm, fill=black!10}}
\newcommand*\circled[1]{~\tikz[baseline=(char.base)]{\node[shape=circle,draw, text width=3.5mm,align=center, minimum size = 3mm, inner sep = 0pt] (char) {\scriptsize \textcolor{black}{\normalsize \textbf{#1}}}}}

\begin{document}
	
	\title{Context-triggered Abstraction-based Control Design}
	
	\author[1*]{S.P.~NAYAK}
	\author[2]{{L.N.}~EGIDIO}
	\author[2]{M.~DELLA ROSSA}
	\author[1]{A.-K.~SCHMUCK}
	\author[2]{R.M.~JUNGERS}
	
	\affil[1]{Max Planck Institute for Software Systems, Kaiserslautern, Germany} 
	\affil[2]{ICTEAM, UCLouvain, Louvain-la-Neuve, Belgium.}

	% \markboth{Context-triggered Abstraction-based Control Design}{S.P.~NAYAK {\itshape ET AL}.}

	\maketitle

	\begingroup\renewcommand\thefootnote{*}
	\footnotetext{First two authors contributed equally. S.~P.~Nayak and A.-K.~Schmuck are partially supported by the DFG project 389792660 TRR 248–CPEC. 
	A.-K.~Schmuck is partially supported by the DFG project SCHM 3541/1-1.
	L.N.~Egidio, M.D.~Rossa, and R.~Jungers are supported by the ERC project under the European Union's Horizon 2020 research and innovation programme under grant agreement No 864017 - L2C.}
	\endgroup

	\begin{abstract}
	We consider the problem of automatically synthesizing a hybrid controller for non-linear dynamical systems which ensures that the closed-loop fulfills an arbitrary \emph{Linear Temporal Logic} specification. Moreover, the specification may take into account logical context switches induced by an external environment or the system itself. Finally, we want to avoid classical brute-force time- and space-discretization for scalability.
	We achieve these goals by a novel two-layer strategy synthesis approach, where the controller generated in the lower layer provides invariant sets and basins of attraction, which are exploited at the upper logical layer in an abstract way. In order to achieve this, we provide new techniques for both the upper- and lower-level synthesis.

	Our new methodology allows to leverage both the computing power of state space control techniques and the intelligence of finite game solving for complex specifications, in a scalable way.
	\end{abstract}
	
	\section{INTRODUCTION}\label{sec:intro}
	%!TEX root = ../main.tex

The problem of synthesizing controllers for different classes of non-linear systems with respect to temporal logic specifications has received considerable attention in the last decades, especially in the context of \emph{cyber-physical systems} (CPS) design. The goal of these methods is to allow for fully automated synthesis of feedback controllers, which enforce temporal logic constraints and hence, to allow for a much larger spectrum of specifications than classical  feedback controller synthesis techniques. In order to achieve this goal, techniques from the formal methods and the control communities need to be combined. 

While there has been enormous progress towards this goal in the last decade, documented by various recent textbooks on this problem, e.g.\ \cite{tabuada2009_book,sanfelice2020hybrid,calin_belta_book}, most of the existing approaches still tackle the overall problem mainly from either the control or the formal methods side. Thereby, the potential of techniques available in the respective other domain is not fully exploited, leading to unsatisfying solutions in settings where low-layer continuous control and high-layer logical decision making are \emph{tightly intertwined.}

Such problems occur for example in the control of autonomous robots deployed in warehouses \cite{amazonRobotics}, under-water inspection \cite{UnderwaterRobots, UnderwaterRobots2} or in rescue and evacuation scenarios \cite{evacuationRobot2,evacuationRobot}. 
% \AKS{todo, someone needs to find good references here. The first one should be on amazone warehouse robots, the second one on this \url{https://eelume.com/}(check out the video, it is very cool!), the person behind this is Kristin Ytterstad Pettersen, the last one, I have no immediate reference in mind.}
In these applications, the robots need to (a) directly compensate \emph{environment uncertainty} during their movement (such as rough terrain or sensor/actuator noise), and (b) strategically react to any \emph{logical context change}, e.g., a newly arriving package that needs to be re-located in the warehouse, a leak in an oil pipeline that needs to be fixed under water, or a door that got closed and needs to be re-opened to reach a target in a rescue scenario. These context changes are triggered by the \emph{external environment} and can occur at any time. They must directly result in (high-level) \emph{strategic reactions} of the robots that trigger new objectives of the (low-level) \emph{feedback control policy} which, on the other hand, is able to correctly actuate  non-trivial non-linear dynamical systems. Control problems with a similar required integration of logical decision making and low-layer feedback control occur for example in sustainable building management \cite{ShaikhNoR}, or smart energy grid operation \cite{SahaJulius2016} or safety-critical medical operations~\cite{WangJung19}. 

This paper presents a novel approach to such integrated control problems, which automatically computes a  \emph{provably correct hybrid controller} that seamlessly reacts to (high-layer) logical context switches. Therein, the main contribution of our work is twofold: the new game-solving formalism we present (i) provides a \emph{certified and reactive interface} between the higher and the lower control layers via \emph{control Lyapunov functions} and (ii) while dismissing \emph{grid-based discretization} of both the input and the state spaces.
% \textcolor{red}{  
On the same line, our approach does not require discretization of time \emph{ab initio}. Rather, it considers time implicitly at the high-level strategy design, and defers the actual discretization of time to the low-level controller design, in an opportunistic way. Thereby, it enhances scalability and avoids numerical problems due to small sampling time intervals. 
% }

 Moreover, the full class of LTL specifications can be considered for a large class of non-linear continuous dynamics.

 % for incorporating relevant information from the lower into the higher-layer synthesis problem.
% 
% by proposing a new technique that integrates automated strategy computation on the higher control layer with low-level feedback control
% 
% . We consider possible changes of the external (logical) \emph{environment}, which are not in control of the user/system, and can thus be interpreted as an input (or disturbance) to which the controller needs to react to. In the underlying continuous control problem, this amounts to drive solutions to a suitable target (while possibly avoiding an unsafe set) corresponding \emph{both} to the logical control specification and the external environment ``state'' at the certain instant of time. In other words, the overall specification requires a (changing) set of targets to be visited in particular ways, depending on time-varying external inputs. 
% 

% 
% In order to concretely illustrate the main challenges of the considered problem class and our main contribution we consider the following a toy example, which will be used for illustration throughout the paper.

% 
% Throughout this paper, we use a toy example variant of these application szenarios to explain 

	\subsection{Motivating Example, Challenges and Contributions}\label{section:motivatingExample}
	%!TEX root = ../main.tex
Throughout this paper, we re-visit the following simple robot control example to outline the challenges and contributions of our new hybrid controller synthesis approach. 

\smallskip
\noindent\textbf{Example.} We consider a simple moving robot $\mathbf{r}$, in a setting composed by two neighboring rooms, connected by a sliding door, as depicted in Fig.~\ref{fig:example}.
There are three target sets: $\cT_1,\cT_2$ in the left room and $\cT_3$ in the right room. An external user (the \emph{environment}), at each instant of time, chooses a mode among $\Mc_i$, $i\in\{1,2,3\}$ indicating the current desired target $\cT_i$ for the robot.  Moreover, the opening status of the door can be controlled by the robot -- entering the target $\cT_1$ or $\cT_3$ opens the door (if it was previously closed) while  entering the target $\cT_2$ closes it (if it was previously open). 
This can be expressed by the LTL formula\footnote{See \cref{subsec:LTL} for an introduction to linear temporal logic (LTL).}
\begin{subequations}\label{equ:intro:spec}
\begin{align}\label{equ:intro:specassump}
% \begin{aligned}
\spec_A = 
% &\globally\left((\Mc_1 \Leftrightarrow \neg \Mc_2 \wedge \neg \Mc_3) \wedge (\Mc_2 \Leftrightarrow \neg \Mc_1 \wedge \neg \Mc_3)\right)\\
&\globally\bigwedge_{1\leq i\leq3}\left( \Mc_i \Leftrightarrow \bigwedge_{1\leq j\neq i\leq 3} \neg\Mc_j\right)\notag\\
\wedge&\globally (\Tc_1 \vee \Tc_3 \Rightarrow \Next \neg \Dc) \wedge \globally (\Tc_2 \Rightarrow \Next \Dc)\notag\\
\wedge &\globally (\Dc \Rightarrow \Dc \weakuntil (\Tc_1 \vee \Tc_3) )
\wedge \globally( \neg \Dc \Rightarrow \neg \Dc \weakuntil \Tc_2 ).
% \end{aligned}
\end{align}
The goal is to design a feedback control policy that reacts to the external environment decisions $\Mc_i$, by moving to the chosen target $\Tc_i$ while adhering to additional safety-constraints, i.e.\ not hitting the walls $\Wc$ (including the door if it is closed). 
This can be expressed by the LTL formula
 \begin{equation}\label{equ:intro:specguarantee}
  \spec_G = \globally\neg \Wc \bigwedge_{i=1,2,3} \left(\finally\globally\Mc_i \Rightarrow \finally\globally \Tc_i\right).
 \end{equation}
 \end{subequations}
%  where 
%  $\Mc_i$'s are interpreted as external propositions signaling the active mode, and $\Tc_i$'s and $\Wc$ are interpreted as observation propositions signaling whether the robot is in the corresponding target set or hitting the walls, respectively.
Summarizing formally, the overall specification for the robot is $\spec_A\Rightarrow\spec_G$, i.e, it needs to guarantee its goal $\spec_G$ while assuming that $\spec_A$ holds. 
\begin{figure}
\centering
\def\svgwidth{0.4\linewidth}
% \includegraphics[width=0.6\linewidth]{figs/example}
%% Creator: Inkscape 1.2.2 (b0a84865, 2022-12-01), www.inkscape.org
%% PDF/EPS/PS + LaTeX output extension by Johan Engelen, 2010
%% Accompanies image file 'example.eps' (pdf, eps, ps)
%%
%% To include the image in your LaTeX document, write
%%   \input{<filename>.pdf_tex}
%%  instead of
%%   \includegraphics{<filename>.pdf}
%% To scale the image, write
%%   \def\svgwidth{<desired width>}
%%   \input{<filename>.pdf_tex}
%%  instead of
%%   \includegraphics[width=<desired width>]{<filename>.pdf}
%%
%% Images with a different path to the parent latex file can
%% be accessed with the `import' package (which may need to be
%% installed) using
%%   \usepackage{import}
%% in the preamble, and then including the image with
%%   \import{<path to file>}{<filename>.pdf_tex}
%% Alternatively, one can specify
%%   \graphicspath{{<path to file>/}}
%% 
%% For more information, please see info/svg-inkscape on CTAN:
%%   http://tug.ctan.org/tex-archive/info/svg-inkscape
%%
\begingroup%
  \makeatletter%
  \providecommand\color[2][]{%
    \errmessage{(Inkscape) Color is used for the text in Inkscape, but the package 'color.sty' is not loaded}%
    \renewcommand\color[2][]{}%
  }%
  \providecommand\transparent[1]{%
    \errmessage{(Inkscape) Transparency is used (non-zero) for the text in Inkscape, but the package 'transparent.sty' is not loaded}%
    \renewcommand\transparent[1]{}%
  }%
  \providecommand\rotatebox[2]{#2}%
  \newcommand*\fsize{\dimexpr\f@size pt\relax}%
  \newcommand*\lineheight[1]{\fontsize{\fsize}{#1\fsize}\selectfont}%
  \ifx\svgwidth\undefined%
    \setlength{\unitlength}{125.60382392bp}%
    \ifx\svgscale\undefined%
      \relax%
    \else%
      \setlength{\unitlength}{\unitlength * \real{\svgscale}}%
    \fi%
  \else%
    \setlength{\unitlength}{\svgwidth}%
  \fi%
  \global\let\svgwidth\undefined%
  \global\let\svgscale\undefined%
  \makeatother%
  \begin{picture}(1,0.85925143)%
    \lineheight{1}%
    \setlength\tabcolsep{0pt}%
    \put(0,0){\includegraphics[width=\unitlength]{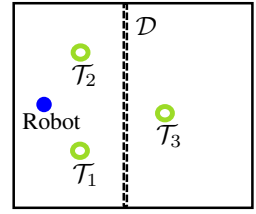}}%
    \put(0.24827397,0.11338743){\color[rgb]{0,0,0}\makebox(0,0)[lt]{\lineheight{1.25}\smash{\begin{tabular}[t]{l}$\Tc_1$\end{tabular}}}}%
    \put(0.59293944,0.27192337){\color[rgb]{0,0,0}\makebox(0,0)[lt]{\lineheight{1.25}\smash{\begin{tabular}[t]{l}$\Tc_3$\end{tabular}}}}%
    \put(0.24320047,0.52378346){\color[rgb]{0,0,0}\makebox(0,0)[lt]{\lineheight{1.25}\smash{\begin{tabular}[t]{l}$\Tc_2$\end{tabular}}}}%
    \put(0.51103159,0.71885604){\color[rgb]{0,0,0}\makebox(0,0)[lt]{\lineheight{1.25}\smash{\begin{tabular}[t]{l}$\Dc$\end{tabular}}}}%
    \put(0.04186939,0.33073131){\color[rgb]{0,0,0}\makebox(0,0)[lt]{\lineheight{1.25}\smash{\begin{tabular}[t]{l}\small{Robot}\end{tabular}}}}%
  \end{picture}%
\endgroup%

\caption{Motivating example: A robot must navigate to and remain at targets $\cT_1$, $\cT_2$ or $\cT_3$ as directed by an external environment which imposes respective modes $\Mc_1$, $\Mc_2$, and $\Mc_3$, while avoiding any collision with the walls $\Wc$ and with the door $\Dc$ (if it is closed).}\vspace{-0.3cm}
\label{fig:example}
\end{figure}

\smallskip
\noindent\textbf{Challenges.} This example showcases three main challenges that are tackled by our new controller synthesis approach. 

\emph{First}, the environment can change the mode \emph{at any time}. Considering a real application where targets might be far away from each other, we would like the robot to immediately adapt its motion towards the new target, and not only after \enquote{completing} the previously assigned task of reaching another target. We achieve this \emph{direct reactivity}, by autonomously switching the low-layer controller in reaction to a mode change. This, however, requires caution to avoid well-known instability problems in switched control settings. 

\emph{Second}, as the robot itself is controlling a part of the logical context (by being able to open and close the door), a hybrid controller cannot naively switch between low-layer controllers for different targets based on the active mode.
% 
% In this example, even if it is possible to design feedback control laws to guide the robot towards the targets $\cT_1,\cT_2,\cT_3$, the interplay between the control specifications and the external context is crucial: 
If, for example, the desired target is set to be equal to $\cT_3$ and the robot is currently in the left room while the door is closed, the robot should automatically decide to first visit the target $\cT_1$ to open the door. Scaling this to applications (e.g., in warehouses) where many logical requirements interact, requires a principled way to design a correct \emph{strategy} for the robot to react to context changes such that a given formal specification, for instance $\spec_A\Rightarrow\spec_G$, is satisfied. 

\emph{Third}, it is important that the low-layer control design  does not simply implement what \emph{should} be done (i.e., which target should be reached) but also what \emph{should not} be done. For example, if the robot is in the left room moving towards $\cT_3$ while the door is open, it must not pass over $\cT_2$, as this would close the door. In addition, the door can be both an obstacle and a target, dependent on the current context.

To design a correct-by-construction hybrid controller tackling the last two challenges, one needs (i) a formally correct mechanism to translate strategic choices from the higher layer to feedback-control problems (with suitable guarantees) in the lower layer and (ii) incorporate all necessary information about the workspace and the low-layer closed-loop properties into the high-layer strategy synthesis problem.

\smallskip
\noindent\textbf{Contribution.} This paper achieves these two goals by a new game-solving formalism for high-layer strategy synthesis, which (i) computes \emph{strategy templates} instead of single strategies and (ii) allows for \emph{progress group augmentations}.
We show that (i) strategy templates provide a \emph{certified top-down interface} by allowing a direct translation into \emph{context-dependent reach-while-avoid} (RWA) controller synthesis problems, which, in turn, can be certifiably solved via control Lyapunov functions. This leads to provably correct low-layer controllers implementing high-layer strategy choices. Further, we show that (ii) progress group augmentations provide a \emph{certified bottom-up interface} that enables a non-conservative and discretization-free incorporation of low-layer closed-loop properties into the higher-layer strategy synthesis game. 

	\subsection{Literature Review}
	%!TEX root = ../main.tex
Existing approaches tackling the outlined integrated controller synthesis problem, can roughly be divided into three different research lines. \emph{First}, discretization-based abstraction techniques can be used to incorporate low-level dynamics into the high-level strategy synthesis games (see e.g., \cite{tabuada2009_book,ReissigWeberRungger_2017_FRR} for an overview and  \cite{Scots,ARCS,Mascot,pFaces,li2018rocs} for tool support). These approaches are able to handle the full problem class we tackle, but are known to suffer heavily from the curse of dimensionality and from conservatism introduced by the abstraction.
\emph{Second}, both the specification and the dynamics of the system can be mapped into a large optimization problem that searches for an optimal control law ensuring that both the logical specification and the dynamical constrains are satisfied (see e.g. \cite{belta2019formal} for a survey). These methods, however, scale poorly with the number of logical constrains and cannot handle external environment inputs.
\emph{Third}, a constrained system can be generated, which searches for certificates on the lower level dynamical system to enforce a temporal specification (see e.g.\ \cite[Ch.12]{sanfelice2020hybrid} for an overview). This approach is usually restricted to particular classes of logical specifications and non-linear dynamics. %To overcome these issues, there have been recent efforts to combine techniques from different approaches  to leverage their individual advantages, e.g. \cite{NilssonAmes_18}. 

Within this paper, we mainly follow the third approach utilizing certificates, in particular control Lyapunov functions, to realize reach-while-avoid objectives. What distinguishes our work from existing ones (e.g., \cite{SahaJulius2016,Jagtap2021,He2020,Xiao2021}) is the presence of logical inputs operated by the external environment. In the absence of these, the resulting synthesis problem reduces to a \emph{temporal logic planning problem}, which does not require a reactive strategy on the higher layer, i.e., a single plan can be computed and executed in an \emph{open-loop fashion}. Our approach produces \emph{closed-loop} controllers in both layers instead.

While recent methods combining certificates with high-granularity abstractions (e.g.\ \cite{NilssonAmes_18}) also produce closed-loop solutions, there, environment inputs can only be handled at transition points between abstract states. In our example, the robot would need to complete one motion (reaching a particular target) before it can receive a new objective, leading to an unsatisfying closed-loop behavior.

In addition, our new game solving formalism is also related to other work in the reactive synthesis community. While strategy templates have been very recently introduced in \cite{strategyTemplatereport,PermissiveAssumptions}, progress group annotations appeared previously in \cite{progress_groups} for a restricted class of temporal specifications and only induced by uncontrolled dynamics. Further, \cite{RaynaRupak2014} also tackles the problem of reactive control for dynamical systems via parity games, but only presents sufficient conditions for the existence of certificates and controllers, while our method is fully constructive.

	\section{PRELIMINARIES}\label{sec:Prelim}
	In this section we recall, in a condensed form,  the main concepts and results from dynamical control systems theory and formal methods settings.
	\subsection{Dynamical Systems}\label{subsec:DynamicalSys}
	%!TEX root = ../main.tex
Let us introduce the state-space setting and the main stabilization/control techniques that we consider in order to achieve the logical specifications described in previous sections. First, we introduce the notion of continuous-time control systems considered in this manuscript.

\begin{definition}\label{defn:ControlSystem}
A (continuous-time) control system is defined by a triple $\cS:=(X,U,f)$ where:
\begin{itemize}[leftmargin=*]
\item the open set $X\subseteq\R^{n_x}$ is the \emph{admissible state space}, of dimension $n_x\in \N$; 
\item the set $U\subseteq \R^{n_u}$ is the \emph{input space}, of dimension $n_u\in \N$;
\item the function $f\in \scrC^1(\R^{n_x}\times \R^{n_u},\R^{n_x})$ describes the \emph{system dynamics}, defined by
\begin{equation}\label{eq:DiffEq}
\dot x=f(x, u).
\end{equation}
	\end{itemize}
\end{definition} 
Given a control system $\cS:=(X,U,f)$ and a measurable function $u: X\to U$, a \emph{solution} of $\cS$ for $u$ starting at $x\in X$ is a function $\sol_{x,u}:[0,T)\to X$  (for some
$T>0$ and possibly $T = +\infty$) such that $\sol_{x,u}(0)=x$, $\sol_{x,u}(t)\in X$ for all $t\in [0,T)$  and $\dot \sol_{x,u}(t)=f\big(\sol_{x,u}(t),u(\sol_{x,u}(t))\big)$ for almost all $t\in [0,T)$. 

To cope with \textit{reach-while-avoid} objectives, we must design control policies driving the solutions to desired targets, possibly avoiding obstacles/staying in safe regions. Thus, we aim to design feedback control strategies, using the formalism of \emph{control Lyapunov functions} (CLF). Let us recall in what follows the main definitions and concepts from CLF-based feedback design literature (for an 
overview, see~\cite{Art83,ClarkeLedRif2000,Clarke11}). To ease notation, we denote by $\scrC^1(X,\R)$ the set of continuously differentiable functions from $X$ to $\R$; given a function $w:X\to \R$ and any $c\in \R$, we denote by $X_w(c):=\{x\in X\;\vert\;w(x)\leq c\}$ the $c$-sublevel set of $w$,  
\begin{definition}\label{def:clf}
Let us consider a compact set $X_T\subset X$ named the \emph{target}. A function $w\in \scrC^1(X,\R)$ is a \emph{control Lyapunov function} (CLF) for system~\eqref{eq:DiffEq} with respect to $X_T$ if there exist $0<c<C$ and $\rho>0$ such that
\begin{align}
 X_{w}(c)&\subseteq X_T\;~\wedge~\;X_{w}(C)\subseteq X , \label{eq:Containement} \\
\inf_{u\in U}\inp{\nabla w(x)}{f(x,u)}&\leq -\rho w(x),\;\forall x\in X_w(C)\setminus X_w(c).\label{eq:Decreasing}
\end{align}
In this case, the set $X_w:=X_w(C)$ is the \emph{basin of attraction} of $w$.
If $X=\R^{n_x}$, $w$ is radially unbounded and inequality~\eqref{eq:Decreasing} holds in $\R^{n_x}\setminus X_w(c)$, then $w$ is said to be a \emph{global} CLF.
\end{definition}
 Intuitively, the condition~\eqref{eq:Decreasing} implies that, whenever $x\in X_w\setminus X_w(c)$, there exists a $u\in U$ for which the directional derivative of $w$ along the vector $f(x,u)$  is strictly negative, and thus the value of the Lyapunov function is decreasing along solutions of~\eqref{eq:DiffEq} following such direction. This observation motivates the following CLF-based design result.
\begin{lemma}\label{lemma:CLFBasedFeedback}
Consider a control system $\cS:=(X,U,f)$, a compact target set  $X_T\subset X$, and suppose that $w\in \scrC^1(X,\R)$ is a CLF in the sense of Definition~\ref{def:clf}.  Consider a continuous $u:X_w\to U$ satisfying
 \begin{equation}\label{eq:u_CLF}
 	\inp{\nabla w(x)}{f(x,u(x))}\leq -\rho w(x),\,\forall x\in X_w\setminus X_w(c),
 \end{equation}
then, for all $x\in X_w$, it holds that $\sol_{x,u}(t)\in X_w$ for all $t\in \R_+$ and $\exists\;T_{x}\geq0$ such that $\sol_{x,u}(t)\in X_w(c),\, \forall t\geq T_{x}$.
\end{lemma}
 The proof follows from classic Lyapunov theory and the comparison argument, therefore, we refer to~\cite{Sontag83,ClarkeLedRif2000} or related literature for a detailed demonstration. 
 
We note that Definition~\ref{def:clf} considers basins of attraction $X_w$ which are sublevel sets of CLFs. Hence, these sets are \emph{safe by construction}, that is, all solutions under a control $u$ satisfying~\eqref{eq:u_CLF} will always stay inside $X_w$ (in addition to eventually reaching $X_w(c)$). As such, the CLFs considered in this paper allow to enforce \emph{reach-while-avoid} objectives, by provably \emph{avoiding} an unsafe region \emph{while} \emph{reaching} a target region within the state space. As the computation of such CLFs can introduce some conservatism, we note that more general approaches, such as control Lyapunov barrier functions (see e.g.~\cite{Xiao2021,Jagtap2021,Clark21}) can similarly be used for the purpose of guaranteeing safety, if no property of convergence is required.
% Also, the fact that $X_w(c)$ is a sublevel set of $w$ ensures that, once a system trajectory enters it, there always exist control actions keeping this trajectory in $X_w(c)$ and thus, by~\eqref{eq:Containement}, in the target $T$. 
% Indeed, the existence of a CLF as given in Definition~\ref{def:clf}, defines the feedback control policy
% \begin{equation}\label{eq:u_CLF}
% 	u(x) \in\arg\min_{u\in U} \nabla w^\top(x) f(x,u)
% \end{equation}
% ensuring the attractiveness and forward-invariance of $X_w(c)$ for trajectories starting in $X_w$.
% Immediately prompting a stabilizing control law, CLFs are powerful tools to design safe controllers.
% 
%
\begin{remark}[CLFs-based Feedback design: Literature review]
	Definition~\ref{def:clf} is stated in a form particularly suited for our purposes and many extensions/modifications are possible.

First of all, let us point out that some technical issues can arise, even in the restricted context of Definitions~\ref{defn:ControlSystem} and~\ref{def:clf}, when considering feedback control laws satisfying~\eqref{eq:u_CLF}. Indeed, functions $u:X\to U$ satisfying~\eqref{eq:u_CLF} can be necessarly discontinuous and thus special care should be provided in defining tailored solution concepts for the closed loop $\dot x=f(x,u(x))$.
For the interested reader, this technical topic is discussed in~\cite[Section 8]{Clarke11}.  In the affine-control case, i.e. when $U=\R^{n_u}$ and $f(x,u)=h(x)+g(x)u$ for some functions $h:\scrC^1(X,\R^{n_x})$ and $g\in \scrC^1(X,\R^{n_x\times n_u})$, a smooth CLF as in Definition~\ref{def:clf} induces a \emph{continuous} feedback law, as defined in~\cite{Sontag89} and well summarized in~\cite{RomJay16}. 
Moreover, for notational simplicity, in Definition~\ref{def:clf} we impose to the candidate CLF the \emph{continuous differentiability} property.
	This hypothesis can be relaxed considering locally Lipschitz candidate control Lyapunov functions.
	In this case, in~\eqref{eq:Decreasing}, Dini-derivatives or Clarke gradient formalism should be used, since the classical gradient is not defined for locally Lipschitz functions.
	We want to stress that, for the classical stabilizability problem of control systems, it is necessary, in order to avoid any conservatism, to consider non-smooth (but locally Lipschitz) CLFs, see~\cite{ClarkeLedRif2000} and references therein. 
	% Another important feature of Definition~\ref{def:clf} is that the decreasing property on the derivative of $w$ in~\eqref{eq:Decreasing} is not imposed inside the $c$-sublevel set, allowing to ensure convergence to a (in general non-singleton and non-connected) compact set.
	% This provides us with a framework rich enough to tackle, with a unique formalism, both \emph{practical} stabilizability, i.e. stabilizability with respect to non-trivial targets, and the classic case of stabilization with respect to an equilibrium point. 
\end{remark}

	\subsection{Linear Temporal Logic}\label{subsec:LTL} 
	%!TEX root = ../main.tex

In this section, we introduce the syntax and semantics of Linear Temporal Logic (LTL) in order to formally describe the logical specifications. For a complete overview, we refer to~\cite[Chapter~5]{BaierKatoen}.

\smallskip
\noindent\textbf{Atomic Propositions.} An atomic proposition is a boolean variable (i.e., a variable that can either be $\true$ or $\false$) which signals important information to the higher-layer logical control layer. In this paper, we consider three different (finite) sets of atomic propositions: 
(i)~\emph{state propositions} $\AP_S$,
(ii)~\emph{observation propositions} $\AP_O$, and
(iii)~\emph{control propositions} $\AP_C$. 
State propositions (e.g., $\Tc_1$, $\Tc_2$, $\Tc_3$ in \cref{fig:example}) are associated with a \emph{subset of the state space} s.t.\ $\cT_i\in \AP_S$ is $\true$ at time $t$ if the current state $x(t)$ of the underlying dynamical system is within this subset\footnote{With a slight abuse of notation we denote the state subset associated with a state proposition by the same symbol.}, i.e. $x(t)\in\cT_i\subseteq X$. 
Observation propositions $\AP_O$ denote all other aggregated information \emph{observed} by the logical controller from the underlying continuous control system (e.g., $\Dc$ in \cref{fig:example}) and the external environment (e.g., $\Mc_1$, $\Mc_2$, and $\Mc_3$ in \cref{fig:example}).
Control propositions $\AP_C$ denote a finite set of feedback control strategies that the high-level logical controller can choose (which will be introduced in \cref{sec:algo:low-to-high}).
% 
% We also use observation propositions to model environment disturbances or partial observation restrictions induced by the abstraction. Intuitively, $\AP_O$ models all influence neither under control of the system nor the logical controller. 
We denote by $\AP:=\AP_S\cup\AP_O\cup\AP_C$ the set of all propositions. % in the system. 

% \begin{example}
% Consider the motivating example given in \cref{fig:example}. As described in \cref{sec:intro}, 
% the propositions $\Tc_1$, $\Tc_2$, $\Tc_3$, and $\Wc$ are state  propositions as they depend uniquely on the robot configuration (i.e., the system state). More precisely, $\cT_i\in\AP_S$ is $\true$ if the current state is in target region $\cT_i$. 
% Furthermore, the propositions $\Mc_1$, $\Mc_2$, and $\Mc_3$ are observation propositions as they are different modes controlled by an external environment. 
% In addition, we have another observation proposition $\Dc$ that signals whether the door is closed, i.e., $\Dc$ is $\true$ if the door is closed, which can be observed by the system. 
% Moreover, we can have a set of control propositions signaling whether the system is using a particular feedback controller.\qed
% \end{example}

Given a control system $\cS=(X,U,f)$, 
the state propositions $\AP_S$ define a labelling function $\labelfunc\colon X\to 2^{\AP_S}$ s.t.\ for all $\Xc\in \AP_S$ holds that $\Xc\in \labelfunc(x)\,\Leftrightarrow\,x\in \Xc$. In addition, $\ldist\colon \R_+\to 2^{\AP_O}$
denotes a piecewise-constant and right-continuous\footnote{A function $L:\R_+\to S$, with $S$ a finite set, is piecewise-constant if  it has a finite number
of discontinuities in any bounded subinterval of $\R_+$; it is right-continuous if $\lim_{s\searrow t} \labelfunc(s)=\labelfunc(t)$ for all $t\in \R_+$.}
\emph{logical disturbance function} modelling the sequence of observation propositions acting on the system over time. We collect all logical disturbance functions acting on $S$ in the set $\Ldist$.
% and
% the atomic propositions in $\AP_O$ observed over time
% can be described by the \emph{labelling functions} 
% $\labelfunc\colon X\to 2^{\AP_S}$ and
% $\labelfunc_O\colon \R_+\to 2^{\AP_O}$, respectively. 
% We assume that, for all  $\Xc\in \AP_S$ we have $\Xc\in \labelfunc(x)\,\Leftrightarrow\,x\in \Xc$ and that $\labelfunc_O$ is 
% We denote by $\labelfunc:=\{\labelfunc,\labelfunc_O\}$ the labelling functions over $\cS$.

% \noindent\textbf{Transition Systems.}
% A transition system over a set of atomic propositions $\AP$ is a tuple $\Aut=(\vertices,M,\Delta,L)$ where $Q$ is a finite set of vertices, $M$ is a finite set of moves, $\Delta: Q\times M \rightarrow 2^Q$ is a non-deterministic transition function, and $L\colon \vertices \rightarrow 2^\AP$ is a vertex labelling function. A run over $\Aut$ is a sequence $\play = q_0q_1\hdots\in Q^\omega$ s.t. for all $i\geq 0$, $q_{i+1}\in\Delta(q_i,m_i)$ for some $m_i\in M$.

\smallskip
\noindent\textbf{Traces.}
For a set $A$, we write $A^\omega$ to denote the set of all infinite sequences $a_0a_1\hdots$ with $a_i\in A$ for each $i\geq 0$.

Then, a \emph{trace} over a set of atomic propositions $\AP$ is an infinite sequence $\trace = l_0l_1\hdots \in (2^\AP)^\omega$. Sometimes we also write $\trace = p_0p_1\hdots \in \AP^\omega $ to denote the trace $\{p_0\}\{p_1\}\hdots$.
% 
% Furthermore, 
Given a control system $\cS$ with labelling function $\labelfunc$, a trace $l_0l_1\hdots$ over $\AP_S\cup\AP_O$ is said to be \emph{generated} by a trajectory $\xi:\R_+\to X$ (of the underlying dynamical system) under disturbance $\ldist\colon \R_+\to 2^{\AP_O}$, if there exists an infinite sequence of time points $\tau_0,\tau_1,\hdots$ for which it holds that:
\begin{itemize}[leftmargin=*]
    \item $\tau_0 = 0$, $\tau_i < \tau_{i+1}$, and $\tau_i$ goes to $\infty$ as $i$ goes to $\infty$,
    
    \item for all $i\in\mathbb{N}$, $t\in [\tau_i,\tau_{i+1})$,  $\,\labelfunc(\xi(t)) \cup \ldist(t) = l_i$ holds.

    %\item for each $i\in\mathbb{N}$, there exists $\tau\in [\tau_i,\tau_{i+1}]$ such that:
   % \begin{itemize}
    %    \item for all $t\in [\tau_i,\tau)$, we have $\labelfunc_O(t)\cup \labelfunc_O(\Phi(t)) = l_i$,
    %    \item for all $t\in (\tau,\tau_{i+1}]$, we have $\labelfunc_O(t)\cup \labelfunc_O(\Phi(t)) = l_{i+1}$,
   %     \item $\labelfunc_O(t)\cup \labelfunc_O(\Phi(t)) = l$ for some $l\in\{l_i,l_{i+1}\}$.
  %  \end{itemize}
\end{itemize}
We write $\traces_{\labelfunc,\ldist}(\xi)$ to denote the set of all traces generated by $\xi$ under $\labelfunc$ and $\ldist$.

\smallskip
\noindent\textbf{Linear Temporal Logic (LTL).} 
We consider requirement specifications written in Linear Temporal Logic~\cite{LTL}. LTL formulas over a set of atomic propositions $\AP$ are given by the grammar
\[\spec \Coloneqq p \mid \specF\vee\specS \mid \neg \spec \mid \Next\specF \mid \specF \until \specS,\] 
where $p\in \AP$ and $\specS$ is an LTL formula. 

A trace $\trace = l_0l_1\hdots \in (2^\AP)^\omega$ is defined to \emph{satisfy} an LTL formula~$\spec$, written as $\trace\vDash\spec$, recursively as follows:
\begin{itemize}[leftmargin=*]
    \item $\trace\vDash p$ if $p\in l_0$;
    \item $\trace\vDash \specF \vee \specS$ if  $\trace\vDash \spec$ or  $\trace\vDash \specS$;
    \item $\trace \vDash \neg \spec$ if $\trace\not\vDash \spec$;
    \item $\trace\vDash\Next \spec$ if $l_1l_2\hdots\vDash \spec$;
    \item $\trace\vDash\specF \until \specS$ if there exists $k\geq 0$ such that $l_il_{i+1}\hdots\vDash \specF$ for all $i< k$ and $l_kl_{k+1}\hdots\vDash\specS$.
\end{itemize}
Furthermore, we define $\true \coloneqq p \vee \neg p$, $\false \coloneqq \neg \true$, and the usual additional operators $\specF \wedge \specS \coloneqq \neg (\neg \specF \vee \neg \specS)$, $\specF \Rightarrow \specS \coloneqq \neg \specF \vee \specS$, 
$\finally \spec \coloneqq \true \until \spec$, $\globally \spec \coloneqq \neg \finally \neg \spec$, and $\specF \weakuntil \specS  \coloneqq (\specF \until \specS) \vee \globally \specF$ for LTL formulas. 
We also use a set of LTL formulas $\{\spec_1,\spec_2,\ldots,\spec_k\}$ as an LTL formula which represents the disjunction of all formulas in it.

\subsection{Games on Graphs}\label{section:games_strategy_template}
In this section, we define the games on graphs and related techniques which will be used to compute a high-level logical controller satisfying  a given LTL specifications. 

\smallskip
\noindent\textbf{Game Graphs.}
A (labelled) \emph{game graph} over a set of atomic propositions $\AP$ is a tuple $\gamegraph= (\vertices,\edges, \labelfuncG)$ consisting of a finite set of \emph{vertices} $ \vertices$ partitioned into two sets: $\p{0}$'s (controller player) vertices and $\p{1}$'s (environment player) vertices, a set of \emph{edges} $ \edges\subseteq \vertices\times\vertices  $, and a labelling function $\labelfuncG\colon \vertices\rightarrow2^\AP$. We write $\verticesi$ to denote $\p{i}$'s vertices, and $\edgesi$ to denote the edges with source in $\verticesi$, i.e., $\edgesi = \edges \cap (\verticesi \times\vertices)$. 
% Moreover, for a set $\vertices'\subseteq\vertices$, we define $\edges(\vertices')$ and $\edgesi(\vertices')$ as the set of end-vertices of the edges $\edges\cap (\vertices'\times\vertices)$ and $\edgesi\cap (\vertices'\times\vertices)$, respectively.
% \mdr{Check with Satya}
% \satya{I couldn't find anyplace (other than the next line) where we use $\edges(V')$, so I removed the definition.)}

% Without loss of generality, we assume that for every $\vertex\in \vertices$ there exists $\vertex'\in \vertices$ s.t. $(\vertex,\vertex')\in \edges$.
A (\p{i}) \emph{dead-end} is a ($\p{i}$) vertex $v$ such that there is no edge from $v$, i.e., $\edges \cap (v\times V) = \emptyset$.  
A \emph{play} from a vertex $\vertex_0$ is a finite or infinite sequence of vertices $\play=\vertex_0\vertex_1\ldots \in \vertices^\omega$ such that $(\vertex_k,\vertex_{k+1})\in \edges$ for all $k\in \N$.

\smallskip
\noindent\textbf{Games.}
A (alternating) two-player game is a pair $\game = (\gamegraph,\textsc{Win})$ consisting of a game graph $\gamegraph=(\vertices,\edges,\labelfuncG)$ such that $\edges \cap (\verticesi\times\verticesi) = \emptyset$ and a winning condition $\textsc{Win}\subseteq V^\omega$. Every winning condition that we consider in this paper can equivalently be expressed as an LTL formula\footnote{We sometimes abuse notation by using the same symbol for the LTL formula and its semantics. An LTL formula $\phi_{\textsc{Win}}$ should not be confused with the control objective $\spec$ over the set $\AP$ defined in \cref{subsec:LTL}. }
$\phi_{\textsc{Win}}$ over a set of propositions interpreted as subsets of $V$ and we use both characterizations interchangeably. %, s.t. $\play\vDash\phi_{\textsc{Win}}$ iff $\play\in\textsc{WIN}$.
A play $\play$ is \emph{winning} if $\play$ ends in a $\p{1}$ dead-end or $\play\in\textsc{Win}$ (or equivalently $\play\vDash\phi_{\textsc{Win}}$). %Note that the labelling function $\labelfuncG$ is not relevant for the game.

% A \emph{strategy} for $\p{i}$, is a function $\strat\colon \vertices^*\verticesi\to \vertices_{1-i}$ such that  
% $(q,\strat(\play q))\in \edges$ holds for every $\play q \in \vertices^*\verticesi$. 
% % 
% Given a  strategy $\strat$ for $\p{i}$, a $\strat$-play is a play $\play=q_0q_1\hdots$ such that $q_{k-1}\in \verticesi$ implies $q_{k} = \strat(q_0\hdots q_{k-1})$ for all $k$.
% % 
% A $\p{i}$ strategy $\strat$ is \emph{memoryless} if $\strat(\play q)=\strat(q)$ holds for every $\play q \in \vertices^*\verticesi$. 
% Hence, such a memoryless strategy for $\p{i}$ can be represented by a function $\strat\colon\verticesi\to\vertices_{1-i}$.
A (memoryless) \emph{strategy} for $\p{i}$, is a function $\strat\colon \verticesi\to \vertices_{1-i}$ such that  
$(v,\strat(v))\in \edges$ holds for every $v \in \verticesi$. 
Given a  strategy $\strat$ for $\p{i}$, a $\strat$-play is a play $\play=v_0v_1\hdots$ s.t.\ $v_{k-1}\in \verticesi$ implies $v_{k} = \strat(v_{k-1})$ for all $k$.

A $\p{0}$ strategy $\strat$ is \emph{winning from a vertex} $\vertex$ if every $\strat$-play from $\vertex$ are winning. 
Moreover, if such a strategy exists for a vertex $\vertex$, then that vertex $\vertex$ is said to be \emph{winning}. We collect all such winning vertices in the \emph{winning region}; and a $\p{0}$ strategy is said to be \emph{winning} if it is winning from every vertex $\vertex$ in the winning region.
Note that we have defined winning strategies only for $\p{0}$ as only $\p{0}$ wants to satisfy the specification in such a (zero-sum) game.

\smallskip
\noindent\textbf{Parity Games.}
A \emph{parity game} is a game with a \emph{parity winning condition} $\paritygame(\priority)$ defined via a priority function $ \priority: V\to [0,d] $ that assigns to each vertex a priority. A play $\play=\vertex_0\vertex_1\ldots $ is winning w.r.t.\ $\paritygame(\priority)$ if the maximum priority seen infinitely often along $\play$ is even. 
The parity winning condition $\paritygame(\priority)$ can be represented by an LTL formula whose atomic propositions are subsets $P_i\subseteq V$ collecting all states with priority $i$, yielding%
% \footnote{We slightly abuse notation by referring to both the LTL formula and its semantics with $\paritygame(\priority)$. Also, $\paritygame(\priority)$ should not be confused with the control objective $\spec$ given for the overall synthesis problem over the set of athomic propositions $\AP$ defined in \cref{subsec:LTL}. }
%, i.e., iff $\max\left(\priority\left(\mathrm{Inf}(\play)\right)\right)$ is even, where $\mathrm{Inf}(\play) := \{\vertex\in V \mid \vertex_k=\vertex \mbox{ for infinitely many }k\in\mathbb{N}\}$.

%Parity specifications are of central importance in the study of synthesis problems as they are general enough to model a huge class of qualitative requirements of cyber-physical systems.

% A parity game is a game $\game=(\gamegraph,\spec)$ with parity specification $\spec = \paritygame(\priority)$ as defined below
\[
    % \paritygame(\priority) \coloneqq
    \bigwedge_{\text{odd }i\in [0;k]} \left(\globally\lozenge P_i \implies \bigvee_{\text{even }j\in [i+1;k]} \globally\lozenge P_j\right).
\]
% where $ P_i=\{\vertex\in \vertices\mid \priority(\vertex)=i \} $ for some priority function $ \priority: V\to [0,d] $ that assigns to each vertex a priority. Intuitively, a play satisfies such a parity specification if the maximum of priorities seen infinitely often is even. 

\smallskip
\noindent\textbf{LTL to Parity Games.}
It is well-known\footnote{We refer the reader to standard textbooks, e.g. \cite{LTLgamesBook}, for more details on LTL, graph games and their connection.} that every LTL formula $\spec$ over some finite proposition set $\AP$ can be translated into an equivalent (labeled) parity game $\game=(\gamegraph,\paritygame(\priority))$. This translation requires a partition of $\AP = \AP_0\cupdot\AP_1$ such that $\p{i}$ (i.e., the controller or the environment player, respectively) chooses the propositions in $\AP_i$. We will see that for the synthesis problems that we consider in this paper, this partition is naturally given.
In addition, plays $\play=\vertex_0\vertex_1\hdots\in V^\omega$ are translated into traces $\trace=l_0l_1\hdots \in (2^\AP)^\omega$ (called \emph{generated} by $\play$) via the labeling function $\labelfuncG$ of $\gamegraph$, s.t.\ 
$l_i = \labelfuncG(\vertex_{2i+1})\cup\labelfuncG(\vertex_{2i+2})$ for each $i\geq 0$. 
Furthermore, we say a game $\game$ or game graph $\gamegraph$ is \emph{total} w.r.t. $\AP'\subseteq \AP$ if for every trace $\trace'$ over $\AP'$, there exists a trace $\trace$ generated by a play in $\gamegraph$ such that $\trace|_{\AP'} = \trace'$.

With this, we recall the following well-known result.
\begin{lemma}[{\cite[Section~4]{spot_ltlsynt}}]\label{lemma:LTL2parity}
Every LTL formula $\spec$ over $\AP = \AP_0 \cupdot \AP_1$ can be translated into a parity game $\game=((\vertices,\edges,\labelfuncG),\paritygame(\priority))$ with $\labelfuncG:=\verticesi\rightarrow 2^{\AP_{1-i}}$ such that $\game$ is total w.r.t. $\AP$.
Moreover, a play is winning in $\game$ iff its generated trace satisfies $\spec$.
\end{lemma}

With \cref{lemma:LTL2parity}, the problem of computing a logical controller which satisfies a given specification $\spec$ in interaction with an uncontrolled environment reduces to computing a winning strategy in a parity game $\game$. 
% . This translation requires, in addition to an LTL formula over $\AP$, a partition of $\AP = \AP_0\cupdot\AP_1$ such that $\p{i}$ controls the propositions in $\AP_i$. 
% The players take turns in the game to perform their moves, which consists in assigning \true or \false to their propositions in order to match the label of the vertex leading to a desired next edge.
% In particular, each vertex of $\p{i}$ is labelled by a set of propositions in $\AP_{1-i}$.
% Hence, given a play, taking the union of the propositions both player choose in every two turns (i.e., one turn of each player) gives us a trace generated by this play. 
% Formally, given a game graph $\gamegraph = (\vertices,\edges,\labelfuncG)$, a trace $\trace=l_0l_1\hdots$ over $\AP$ is said to be \emph{generated} by a play $\play=\vertex_0\vertex_1\hdots$ if $l_i = \labelfuncG(\vertex_{2i})\cup\labelfuncG(\vertex_{2i+1})$ for each $i\geq 0$. 
% We refer the reader to standard textbooks, e.g. \cite{LTLgamesBook}, for more details on LTL, graph games and their connection.
% Moreover, the translation from an LTL formula to a parity game can be formalized as follows.

\subsection{Strategy Templates}\label{sec:Prelim:StrategyTemplates}
% \smallskip
% \noindent\textbf{Strategy templates.}
While it is well known how to compute a \emph{single} winning strategy for a parity game $\game$, it was recently shown that \emph{strategy templates}~\cite{strategyTemplatereport}, which characterize an infinite number of winning strategies in a succinct manner, are particularly useful in the context of CPS control design. They are utilized within this paper to obtain a novel translation of high-level logical control actions into low-level feedback controllers.

\emph{Strategy templates} are constructed from three types of local edge conditions, i.e., \emph{safety}, \emph{co-live} and \emph{live-group} templates.
Formally, given a game $\game=(\gamegraph = (\vertices,\edges,\labelfuncG),\textsc{Win})$, a strategy template is a tuple $(\safegroup,\colivegroup,\livegroup)$ consisting of a set of \emph{unsafe} edges $\safegroup\subseteq \edgeso$, a set of \emph{co-live} edges $\colivegroup\subseteq \edgeso$, and a set of live-groups $\livegroup \subseteq 2^{\edgeso}$. This strategy template can also be represented by an LTL formula $\assump = \assumpsafe(\safegroup) \wedge \assumpdep(\colivegroup) \wedge \assumpgrlive(\livegroup)$, where
\begin{align*}
    &\assumpsafe(\safegroup) \coloneqq \bigwedge_{e\in \safegroup} \globally\neg e,\\
    &\assumpdep(\colivegroup) \coloneqq \bigwedge_{e\in \colivegroup} \finally\globally\neg e,~\text{and}\\
    &\assumpgrlive(\livegroup) \coloneqq \bigwedge_{\livegroupSingleN\in\livegroup} \globally\finally \src(\livegroupSingleN) \Rightarrow \globally\finally \livegroupSingleN.
\end{align*}
Here, an edge $e = (u,v)$ represents the LTL formula $u\wedge\Next v$, and $\src(\livegroupSingleN)$ is the source set $\{\vertex_1,\vertex_2,\ldots,\vertex_k\}$ of the edges in the  live-group $\livegroupSingleN = \{(\vertex_1,\vertex_1'),(\vertex_2,\vertex_2'),\ldots,(\vertex_k,\vertex_k')\}\in\livegroup$. 

A $\p{0}$'s strategy $\strat$ satisfies a strategy template $\assump$ if it is winning in the game $(\gamegraph,\assump)$.
Intuitively, $\p{0}$'s strategy $\strat$ satisfies a strategy template $(\safegroup,\colivegroup,\livegroup)$ if every $\strat$-play $\play$ satisfies the following: 
\begin{enumerate}[label=(\roman*)]
    \item\label{item:strategyTemplate:1} $\play$ never uses the unsafe edges in~$\safegroup$;
    \item\label{item:strategyTemplate:2} eventually, $\play$ stops using the co-live edges in~$\colivegroup$; and
    \item\label{item:strategyTemplate:3} if $\play$ visits $\src(\livegroupSingleN)$ infinitely many times, then it also uses the edges in $\livegroupSingleN$ infinitely many times.
\end{enumerate}

Moreover, a strategy template $\assump$ is \emph{winning} if every strategy satisfying $\assump$ is winning in the original game~$\game$.
Note that sources of all the edges in these templates are $\p{0}$'s vertices. %, and hence, such winning strategy templates represents a class of (potentially) infinitely many winning strategies by local edge conditions~\ref{item:strategyTemplate:1}-\ref{item:strategyTemplate:3}.
The algorithm to compute a winning strategy template in a parity game lies in same time complexity class as the standard algorithm, i.e., Zielonka's algorithm~\cite{Zielonka98}, for solving parity games. This leads to the following result:
\begin{lemma}[\cite{strategyTemplatereport}]\label{lemma:strategytemplate}
Given a parity game with game graph $\gamegraph=(\vertices,\edges,\labelfuncG)$ and priority function $\priority\colon\vertices\to[0,d]$, a winning strategy template can be computed in $\mathcal{O}\left(\abs{\vertices}^{d+\mathcal{O}(1)}\right)$ time.
\end{lemma}

% We refer the reader to the original paper on strategy template~\cite{strategyTemplatereport} for more details.

	\section{PROBLEM STATEMENT}\label{sec:ProblemStatement}
	%!TEX root = ../main.tex

This section gives a formal definition of the problem we are tackling in this paper.
Our goal is to automatically synthesize a reactive hybrid controller that operates a non-linear control system based on external logical inputs. %The controller is designed to meet a specific requirement expressed in Linear Temporal Logic (LTL) using a finite set of logical input and observation predicates.	Unlike methods that require time and space discretization, our algorithm operates directly on the continuous model of the system. To do so, it creates a game graph that incorporates all the essential information about the underlying dynamics. This will be clarified in what follows but, first, let us formalize the main problem to be solved.
Towards a formal problem statement, we first define \emph{a hybrid state-feedback control policy} which controls a system $\cS$ while reacting to logical context switches induced by the sequence of observation propositions $\ldist\in\Ldist$ acting on $\cS$ as logical disturbances.

% \begin{definition}\label{defn:SolutionsHybrid}
%  Let $\cS=(X,U,f)$ be a control system with labelling functions $\labelfunc:X\to 2^{\AP_S}$ and  
% $\labelfunc_O\colon \R_+\to 2^{\AP_O}$. A \emph{hybrid or context-dependent state-feedback policy} is a function $p:X\times2^{\AP_S}\cup 2^{\AP_O}\to U$. A \emph{solution} of $\cS$  for $p$ starting at $x\in X$ is a function\footnote{If $x$ and $p$ are clear from the context they are omitted for simplicity.} $\sol_{x,p}:[0,T)\to X$  (for some
% $T>0$ and possibly $T = +\infty$) such that $\sol_{x,p}(0)=x$, $\sol_{x,p}(t)\in X$ for all $t\in [0,T)$ and $\dot \sol_{x,p}(t)=f(\sol_{x,p}(t),p(\sol_{x,p}(t),\labelfunc(\sol_{x,p}(t))\cup\labelfunc_O(t))\,)$ for almost all $t\in [0,T)$. 
% \end{definition}
% \AKS{removed $\labelfunc$ to fit fig.\ref{Fig:BlockDiagram}.}
\begin{definition}\label{defn:SolutionsHybrid}
 Let $\cS=(X,U,f)$ be a control system and $\ldist\colon \R_+\to 2^{\AP_O}$ a disturbance function. A \emph{hybrid state-feedback policy} is a function $p:\R_+\times X\times \Ldist\to U$. A \emph{solution} of $\cS$ for $p$ starting at $x\in X$ under $\ldist$ is a function % \footnote{If $x$ and $p$ are clear from the context they are omitted for simplicity.} 
 $\sol_{x,p,\ldist}:[0,T)\to X$  (for some
$T>0$ and possibly $T = +\infty$) such that $\sol_{x,p,\ldist}(0)=x$, $\sol_{x,p,\ldist}(t)\in X$ for all $t\in [0,T)$ and $\dot \sol_{x,p,\ldist}(t)=f(\sol_{x,p,\ldist}(t),p(t,\sol_{x,p,\ldist}(t),\ldist(t))\,)$ for almost all $t\in [0,T)$. 
\end{definition}

This leads us to the following problem statement. 

\begin{problem}\label{prob:MainProb}
Given a control system $\cS=(X,U,f)$ with labelling function $\labelfunc\colon X\to 2^{\AP_S}$ and an LTL specification $\spec$ over the predicates $\AP_S\cup\AP_O$, find a set of \emph{winning initial conditions}  $X_{\text{win}}\subseteq X$ and  hybrid state-feedback policy $p:\R_+\times X\times \Ldist\rightarrow U$ s.t.\ 
for all $x \in X_{\text{win}}$, all disturbance functions $\ldist\in\Ldist$ and all solutions $\sol_{x,p,\ldist}$, it holds that
% $u(t)=p(\Phi(t),\labelfunc_O(t)\cup\labelfunc_O(t))$ ensures:
\begin{enumerate}
	\item[(i)] $\sol_{x,p,\ldist}(t)\in X_{\text{win}}$ for all $t\in\R_+$, and
	\item[(ii)] every trace $\trace \in \traces_{\labelfunc,\ldist}(\sol_{x,p,\ldist})$ satisfies $\spec$.
\end{enumerate}
\end{problem}
The remainder of this paper illustrates our solution to Problem~\ref{prob:MainProb} by first providing an overview of the entire multi-step synthesis algorithm in \cref{sec:ControlStrategy}, then highlighting additional details for selected steps in~\cref{subsec:DiscreteLayer} and~\cref{sec:ImplementDetails}, and showing simulation results for the motivating example from \cref{section:motivatingExample} in \cref{sec:RunningExample}.
% \AKS{this is not correct yet, because you have not defined trajectories defined by feedback controllers. You only define it for a fixed control input function. In the previous section you give $x$ and $u$ and define $\Phi_{x,u}$. In the previous version of this problem statement you defined $u$ (i.e., $p$) via $\Phi$, which is not even defined if there was no $u$ before. Can you please fix this?}\mdr{I redefined solutions in Section II for time-varying feedbacks laws}
% \AKS{also, why is it important that $\cS$ has an initial state? If this is important we need to ensure that $x_0\in X_w$. No?}\mdr{No, indeed we don't need initial conditions, we just want to find a set of winning initial states ($X_w$)!....if you agree, feel free to delete these comments}

    \section{SYNTHESIS OVERVIEW}\label{sec:ControlStrategy}
	%!TEX root = ../main.tex
This section overviews our automated synthesis procedure which consists of five steps which are schematically depicted in \cref{fig:overview}. First, in \cref{sec:algo:high} (\cref{fig:overview}, green) we solve a high-level logical game induced by the specification. Then, in \cref{sec:algo:high-to-low} (\cref{fig:overview}, pink) we build a \emph{top-down} interface which allows us to translate strategic choices from the logical level into certified low-level feedback control policies. Afterwards, in \cref{sec:algo:low-to-high} (\cref{fig:overview}, cyan), we build a \emph{bottom-up} interface to include relevant information about the low-level closed-loop dynamics into the logical synthesis game via \emph{augmentations}. We then solve the resulting \emph{augmented} high-level synthesis game in \cref{subsec:solveaugmentedgame} (\cref{fig:overview}, violet). Finally, in \cref{subsec:hybridcontroller} (\cref{fig:overview}, orange), the obtained winning strategy is used to construct a hybrid controller which is proven to solve \cref{prob:MainProb}.

\begin{figure*}
    \scriptsize
    \centering
    \begin{tikzpicture}[node distance=1.5cm]
        % Nodes
        \node[box,ultra thick] (v1) {\circled{1}\begin{tabular}{c}logical specification\\$\spec$\end{tabular}};
        \node[box] (v2) [right=1.8cm of v1] {\circled{2}\begin{tabular}{c}initial parity game\\$\game^I$\end{tabular}};
        \node[box] (v3) [below=1cm of v2] {\circled{3}\begin{tabular}{c}strategy template\\$\assump^I$\end{tabular}};
        \node[box] (v4) [below=1cm of v3] {\begin{tabular}{c}reach-while-avoid objectives\\\circled{4} $\allRWA$\end{tabular}};
        \node[boxl] (v5) [below=1cm of v4] {\begin{tabular}{c}control Lyapunov functions $\allCLF$\\\circled{5} control policies $\allU$\end{tabular}};
        \node[boxl,ultra thick] (v0) [below=5 cm of v1] {\circled{0}\begin{tabular}{c}control system\\$\cS=(X,U,f)$\end{tabular}};
        \node[box] (v6) [right=1.8cm of v2] {\circled{6}\begin{tabular}{c}merged game\\ $\game^M$\end{tabular}};
        \node[box] (v7) [right=of v4] {\circled{7}\begin{tabular}{c}control graph\\ $\gamegraph^C$\end{tabular}};
        \node[box] (v8) [right=1cm of v7] {\circled{8}\begin{tabular}{c}persistent live-groups\\ $\perslivegroups^C$\end{tabular}};
        \node[box] (v9) [right=2.5cm of v6] {\circled{9}\begin{tabular}{c}final augmented game\\ $\game^F$\end{tabular}};
        \node[box] (v10) [below=1cm of v9] {\circled{10}\begin{tabular}{c}winning region $\win^F$\\ winning strategy $\strat^F$\end{tabular}};
        \node[boxl,ultra thick,accepting] (v11) [below=3cm of v10] {\circled{11}\begin{tabular}{c}winning initial conditions  $X_{\text{win}}$\\  hybrid state-feedback policy $p$\end{tabular}};

        % paths
        \draw[->, ultra thick] (v1) -- (v2) node[fill=green!70!black!30,inner sep=3pt, midway, above=0.1cm,rounded corners] {\hyperref[sec:algo:high]{Sec.~\ref*{sec:algo:high}}};
        \draw[->, ultra thick] (v2) -- (v3) node[fill=green!70!black!30,inner sep=3pt, midway, left=0.1cm,rounded corners] {\hyperref[sec:algo:high]{Sec.~\ref*{sec:algo:high}}};

        \draw[->, ultra thick] (v3) -- (v4) node[fill=pink,inner sep=3pt,midway, left=0.1cm,rounded corners] {\hyperref[step3]{Sec.~\ref*{step3}}};
        \draw[->, ultra thick] (v4) -- (v5) node[fill=pink,inner sep=3pt,midway, left=0.1cm,rounded corners] {\hyperref[step4]{Sec.~\ref*{step4}}};
        \draw[->, ultra thick] (v0) -- (v5);

        \draw[->, ultra thick] (v2) -- (v6) node[fill=cyan!20,inner sep=3pt,midway, above=0.1cm,rounded corners] {\hyperref[step4a]{Sec.~\ref*{step4a}}};
        \draw[->, ultra thick] (9.35cm,-6.05cm) -- (v7.south) node[fill=cyan!20,inner sep=3pt,midway, left=0.1cm,rounded corners] {\hyperref[step4b]{Sec.~\ref*{step4b}}};
        \draw[->, ultra thick] (12.25cm,-6.05cm) -- ([xshift=-1cm]v8.south) node[fill=cyan!20,inner sep=3pt,midway, left=0.1cm,rounded corners] {\hyperref[step5]{Sec.~\ref*{step5}}};
        \draw[->, ultra thick] (v5) -- (v11);

        \draw[->, ultra thick] (v6) -- (v9) node[fill=cyan!20,inner sep=3pt,midway, above=0.1cm,rounded corners] {\hyperref[step6]{Sec.~\ref*{step6}}};
        \draw[-, ultra thick] ([xshift=-1cm]v8.north) -- ([xshift=-0.95cm]v9.west);
        \draw[-, ultra thick] (v7.north) |- ([xshift=-1cm,yshift=1cm]v8.north);

        \draw[->, ultra thick] ([xshift=0.5cm]v9.south) -- ([xshift=0.5cm]v10.north) node[fill=violet!30,inner sep=3pt,midway, left=0.1cm,rounded corners] {\hyperref[subsec:solveaugmentedgame]{Sec.~\ref*{subsec:solveaugmentedgame}}};
        \draw[->, ultra thick] ([xshift=0.5cm]v10.south) -- ([xshift=0.5cm]v11.north) node[fill=orange!60,inner sep=3pt,midway, left=0.1cm, yshift=-1cm,rounded corners] {\hyperref[subsec:hybridcontroller]{Sec.~\ref*{subsec:hybridcontroller}}};

    \end{tikzpicture}
 \caption{Flowchart illustrating the overall algorithm given in Section~\ref{sec:ControlStrategy}. Nodes \protect\circled{0},\protect\circled{1} are the inputs and node \protect\circled{11} is the output of our synthesis method. 
  High-level and low-level synthesis steps are colored in dark and light grey, respectively, and discussed in the sections indicated at the arrows.
 }\label{fig:overview}
\end{figure*}
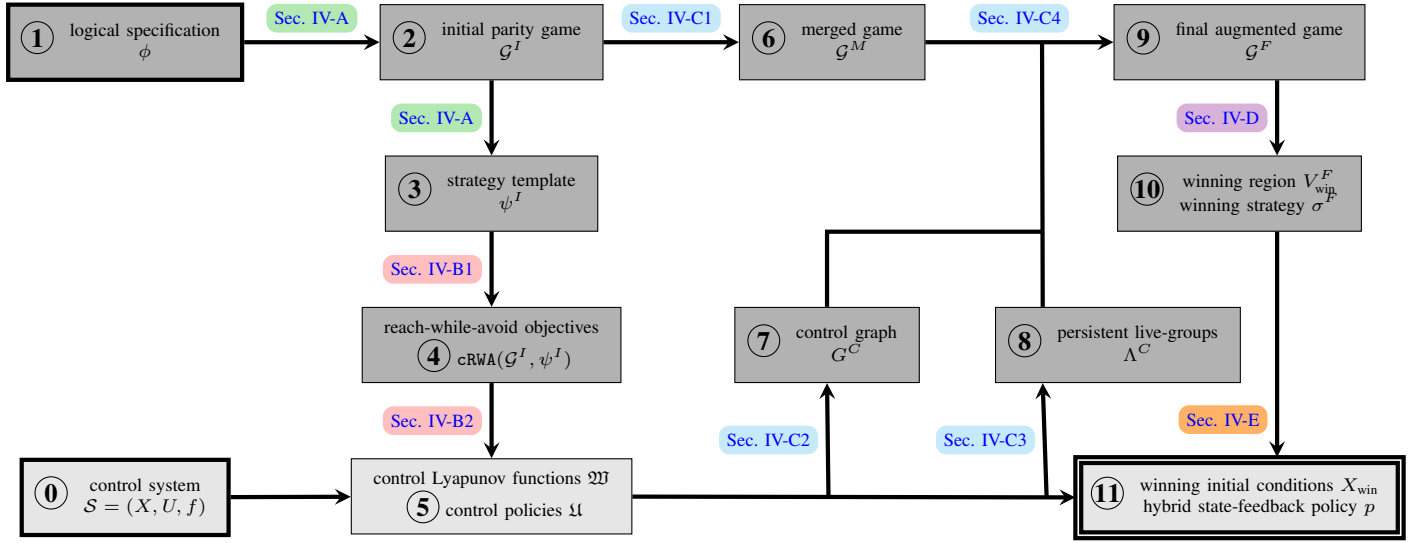

\subsection{High-Level Logical Synthesis}\label{sec:algo:high}
%!TEX root = ../main.tex    
This initial step only considers the (high-level) logical strategy synthesis problem induced by the LTL specification~$\spec$ (realizing the green marked transitions in \cref{fig:overview}). As formalized in \cref{prob:MainProb}, the specification $\spec$ only contains state and observation propositions, i.e., $\AP=\AP_S\cup\AP_O$. The definition of control propositions $\AP_C$ is part of our synthesis framework and will be discussed in \cref{sec:algo:high-to-low}. 

In order to use \cref{lemma:LTL2parity} to construct the \emph{initial parity game $\game^I$} from $\spec$, we need to divide $\AP$ into controller (player $0$) and environment (player $1$) propositions. To do this, we optimistically assume that the controller can instantly activate/deactivate all state propositions in $\AP_S$, thus defining \ $\AP_0 := \AP_S$. This ignores the dynamics of $\cS$ and how the state propositions are geometrically represented in the state-space. 
This is done on purpose to enable a \emph{lazy} synthesis framework -- our framework only adds aspects of both the dynamics and the geometric constraints which show to be \emph{relevant} to the synthesis problem in a later step, discussed in \cref{sec:algo:low-to-high}. 

As observation propositions are not under the control of the system or the controller, they are naturally interpreted as environment propositions, i.e., $\AP_1 := \AP_O$. 
Intuitively, the initial game $\game^I$ constructed from $\spec$ via \cref{lemma:LTL2parity} reveals all logical dependencies of propositions relevant to the synthesis problem at hand. After constructing $\game^I$ from $\spec$ (i.e., going from \circled{1} to \circled{2} in \cref{fig:overview}), we can directly apply the algorithm from \cite{strategyTemplatereport} to synthesize a winning strategy template $\assump^I$ on $\game^I$ (i.e., going from \circled{2} to \circled{3} in \cref{fig:overview}) as discussed in \cref{sec:Prelim:StrategyTemplates}.

This gives the following result which is a direct consequence of \cref{lemma:LTL2parity} and the definition of strategy templates.

\begin{proposition}\label{prop:result:high}
Given the LTL specification $\spec$ over $\AP=\AP_S\cup\AP_O$ translated into an initial parity game $\game^I$ that is total w.r.t. $\AP$ via \cref{lemma:LTL2parity} and a winning strategy template $\assump^I$  for $\game^I$ the following holds: 
for every $\p{0}$ strategy $\strat$ that satisfies the strategy template $\assump^I$, it holds that the trace generated by a $\strat$-play in the initial game $\game^I$ satisfies the specification~$\spec$.
\end{proposition}

\begin{example}\label{ex:strategyTemplate}
For the example from \cref{section:motivatingExample}, the parity game $\game^I$ is constructed from the LTL specification $\spec = \spec_A\Rightarrow\spec_G$ in \eqref{equ:intro:spec} using \cref{lemma:LTL2parity} 
with $\AP_0 = \{\Tc_1,\Tc_2,\Tc_3,\Wc\}$ and $\AP_1 = \{\Mc_1,\Mc_2,\Mc_3,\Dc\}$. 
A part of the resulting parity game $\game^I$ is depicted in Fig.~\ref{fig:gamegraph}.

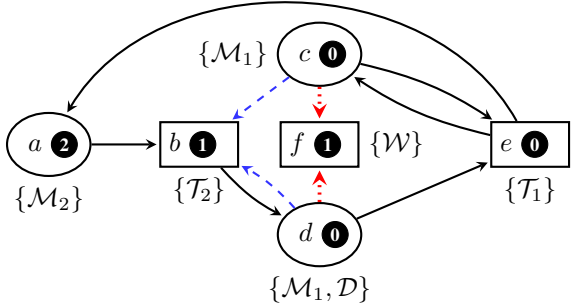
\begin{figure}[b]
\centering
\begin{tikzpicture}
	\node[bplayer0,label={[align=center]below:$\{\Mc_2\}$}] (a) at (0, 0) {$a$\bcircled{2}};
	\node[bplayer1,label={[align=center]below:$\{\Tc_2\}$}] (b) at (\bhpos,0) {$b$\bcircled{1}};
	\node[bplayer0,label={[align=center]left:$\{\Mc_1\}$}] (c) at (1.8*\bhpos, \bypos) {$c$\bcircled{0}};
	\node[bplayer0,label={[align=center]below:$\{\Mc_1,\Dc\}$}] (d) at (1.8*\bhpos, -\bypos) {$d$\bcircled{0}};
	\node[bplayer1,label={[align=center]below:$\{\Tc_1\}$}] (e) at (3.2*\bhpos, 0) {$e$\bcircled{0}};
	\node[bplayer1,label={[align=center]right:$\{\Wc\}$}] (f) at (1.8*\bhpos, 0) {$f$\bcircled{1}};

	\path[->] (a) edge (b);
	\path[->] (b) edge[bend left =-10] (d);
	\path[->] (c) edge[dashed,blue!80] (b) edge[bend left =10] (e) edge[dotted,red, line width = 0.05cm] (f);
	\path[->] (d) edge (e) edge[dotted,red, line width = 0.05cm] (f) edge[dashed,blue!80,bend left =-10] (b);
	\path[->] (e) edge[bend left =10] (c) edge[bend left =-60] (a);
\end{tikzpicture}
\caption{Illustration of a part of the \emph{initial parity game} for the motivating example with $\p{1}$ (squares) vertices and $\p{0}$ (circles) vertices containing their priority in a black circle.
A winning strategy template consists of unsafe edges indicated by red dotted arrows and co-live edges indicated by blue dashed arrows.}\label{fig:gamegraph}
\end{figure}

A winning strategy template for the part of the parity game $\game^I$ depicted in \cref{fig:gamegraph} is 
\[\assump^I = \assumpsafe(e_{cf},e_{df}) \wedge \assumpdep(e_{cb},e_{db}),\]
where $e_{vv'}$ denotes the edge from $v$ to $v'$.

The strategy template $\assump^I$ forces the plays to never use the unsafe edges $\{e_{cf},e_{df}\}$ (indicated schematically by dotted red arrows) as they lead to vertex~$f$ where proposition $\Wc$ is true signaling that the robot hits the wall.
Furthermore, $\assump^I$ forces the plays to eventually stop using the co-live edges $\{e_{cb},e_{db}\}$ (indicated schematically by dashed blue arrows). This is because if $\p{0}$ (i.e., the controller) keeps using these edges, then $\p{1}$ (i.e., the environment) can force a play to loop in one of the cycles $(cbde)^\omega$ or $(db)^\omega$ which does not lead to a winning play as the maximum priority seen infinitely often is odd (i.e., $1$) in these cycles.\qed
\end{example}

\subsection{The Top-Down Interface}\label{sec:algo:high-to-low}
%!TEX root = ../main.tex
While \cref{sec:algo:high} utilizes existing techniques from reactive synthesis, this section contains the first technical contribution of the paper which is the translation of strategy templates into certified low-level feedback control policies (realizing the pink marked transitions in \cref{fig:overview}).

\subsubsection{Reach-While-Avoid-Objectives}\label{step3}
The strategy template $\assump^I$ computed in the last step defines, for all $\p{0}$ vertices $v$, eventually required transitions (contained in $\livegroup$) and (eventually) prohibited transitions (contained in $\safegroup$ or $\colivegroup$) for strategies that result in a correct closed-loop behavior.
While the game solving engine assumes that these transitions can be instantaneously enabled (resp. disabled), they actually have to be enforced (resp. prevented) by a suitable actuation of the underlying dynamical system (e.g., the robot). The main observation that we exploit in this step is that the edge constraints for a $\p{0}$ vertex $v$ induced by a strategy template $\assump^I$ naturally translate into \emph{context-dependent reach-while-avoid objectives} for the lower-layer synthesis problem. %, as defined next.

\begin{definition}\label{def:rwa}
    A \emph{context-dependent reach-while-avoid objective} (cRWA) is defined as a triple $\RWA := (\contextRWA,\reachRWA,\avoidRWA)$  where $\contextRWA\subseteq \AP_O$ is the \emph{context}, $\reachRWA\in 2^{\AP_S}$ is the \emph{target set} (to be reached) and $\avoidRWA\in 2^{\AP_S}$ is the \emph{obstacle set} (to be avoided). 
    A control proposition $\Cc\in\AP_C$ is said to \emph{implement} the  \emph{reach-while-avoid objective} $\RWA$ if %and only if
	\begin{equation}
		\spec_\Cc:=\globally(\square (\Cc ~\wedge~\contextRWA) \Rightarrow \Diamond\square \reachRWA ~ \wedge ~  \square \neg\avoidRWA).\label{eq:feas_control}
	\end{equation}
\end{definition}

In practice, the translation of winning strategy templates into reach-while-avoid objectives(i.e., going from \circled{3} to \circled{4} in \cref{fig:overview}) is done per vertex $v\in V_0$ (whose label defines the context) and reflects required and prohibited successors as targets and obstacles in the cRWA, respectively.
In particular, as the final hybrid controller will make strategic decisions corresponding to exactly one transition, we compute cRWA's per required/allowed transition, while collecting all prohibited successors in the obstacles $\avoidRWA$ of these cRWA's, as formalized next.

\begin{definition}\label{def:generatedRWA}
 Let $\game$ be a parity game with game graph $\gamegraph=(V,E,\labelfuncG)$ and winning strategy template $\assump=(\safegroup,\colivegroup,\livegroup)$. For every $\vertex\in V_0$ let $\mathrm{Suc}_\reachRWA(\vertex)=\{\altvertex\in V_1~|~(\vertex,\altvertex)\notin\safegroup\cup\colivegroup\}$. 
 Then, for each $\altvertex\in\mathrm{Suc}_\reachRWA(\vertex)$ we define
 $\RWAa(\vertex,\altvertex) := (\labelfuncG(\vertex),\labelfuncG(\altvertex),\avoidRWAa(\vertex))$ and
 $\RWAe(\vertex,\altvertex) := (\labelfuncG(\vertex),\labelfuncG(\altvertex),\avoidRWAe(\vertex))$
 s.t.\ 
 \begin{itemize}
	\item $\avoidRWAa(\vertex)=\bigcup_{\{\vertex''\in V_1~|~(\vertex,\vertex'')\in\safegroup\}}\labelfuncG(\vertex'')$, and
    \item $\avoidRWAe(\vertex)=\bigcup_{\{\vertex''\in V_1~|~(\vertex,\vertex'')\in\safegroup\cup\colivegroup\}}\labelfuncG(\vertex'')$.
 \end{itemize}
We collect all such cRWA's for the strategy template $\assump$ in the set $\cRWA(\game,\assump)$.
\end{definition}
Intuitively, for such cRWA's, $\avoidRWAa$ consists of the propositions that need to be avoided ``always'', whereas $\avoidRWAe$ consists of the propositions that need to be avoided ``eventually always''.
This definition is illustrated by the follwing example.

\begin{example}\label{ex:RWAs}
Consider the winning strategy template $\assump^I$ computed in \cref{ex:strategyTemplate} for the parity game given in \cref{fig:gamegraph}. 
From vertex~$d$, strategy template $\assump^I$ forces $\p{0}$ to never use edge $e_{df}$ and eventually stop using edge $e_{db}$. That means, $\p{0}$ has to eventually only use edge $e_{de}$ from vertex~$d$. 
The labels of the vertices imply that whenever mode $\Mc_1$ is active and the door is closed, the system ``always'' has to reach $\cT_1$ while avoiding walls $\Wc$ and ``eventually always'' has to reach $\cT_1$ while avoiding both walls $\Wc$ and target $\Tc_2$.
This leads to the cRWA's $\RWAa(d,e) =(\labelfuncG(d),\labelfuncG(e),\avoidRWAa(d))$ and $\RWAe(d,e) =(\labelfuncG(d),\labelfuncG(e),\avoidRWAe(d))$, where $\labelfuncG(d) = \{\Mc_1,\Dc\}$, $\labelfuncG(e)=\{\Tc_1\}$, $\avoidRWAa(d) = \{\Wc\}$, and $\avoidRWAe(d) = \{\Wc,\Tc_2\}$.\qed
\end{example}

\subsubsection{Feedback-Control Policies} \label{step4}
Within this step, we utilize existing techniques to synthesize a feedback-control policy $u: X\to U$ associated to cRWA problem $\RWA= (\contextRWA,\reachRWA,\avoidRWA)$ (i.e., going from \circled{4} to \circled{5} in \cref{fig:overview}), s.t.\ all traces generated by solutions of $\cS$ for $u$ satisfy \eqref{eq:feas_control}, given that $\Cc$ and $\kappa$ are true for all  $t\in\R_+$, where $\Cc\in\AP_C$ is a controller proposition that flags that the feedback control policy $u$ is currently applied to $\cS$. 
This part of our controller design strategy comes with unavoidable conservatism. Indeed, it is well-known that very particular cases of the control problems that we tackle here face strong controllability barriers, such as undecidability and NP-hardness (see~\cite{BloTsi2000}). For this reason, we rely here on control techniques that are intrinsically conservative, but provide, when they converge, a satisfactory solution.

As an example of such approaches, which fits particularly well with our setting, we utilize existing techniques based on control Lyapunov functions (CLF), as introduced in \cref{subsec:DynamicalSys}, to define $u$ from an $\RWA = (\contextRWA,\reachRWA,\avoidRWA)$. This is achieved by constructing a CLF
$w:X\to \R$ (recall Definition~\ref{def:clf}) w.r.t.\ to the target $\reachRWA$ and enforcing that the basin of attraction $X_w\subseteq X$ excludes $\avoidRWA$, i.e.\ $\avoidRWA\cap X_w=\emptyset$.

We thus have the following definition.
\begin{definition}\label{def:RWAtoPolicy}
Given the control system $\cS=(X,U,f)$, consider a cRWA $\RWA= (\contextRWA,\reachRWA,\avoidRWA)$. We say that a CLF $w$ (as in \cref{def:clf}) with basin of attraction $X_w$ and the corresponding feedback map $u_w:X_w\to U$ satisfying conditions in \cref{lemma:CLFBasedFeedback} are \emph{associated to $\RWA$} if $X_w\cap \avoidRWA=\emptyset$ and $X_w(c)\subseteq \reachRWA$.
\end{definition}

\cref{subsec:ContinuousLayer} will discuss a particular technique to synthesize $X_w$ and $u_w$ realizing a cRWA for particular classes of dynamical systems and state propositions. For any such realization of a cRWA  we have the following guarantees on the resulting closed-loop system under a constant context, i.e., w.r.t.\ a trivial distrubance function $\Upsilon:=\kappa^\omega$, which are a direct consequence of \cref{lemma:CLFBasedFeedback} and \cref{def:RWAtoPolicy}.

\begin{proposition}\label{prop:result:high-to-low}
Given the control system  $\cS=(X,U,f)$ with labelling function $\labelfunc$, let $\RWA= (\contextRWA,\reachRWA,\avoidRWA)$ be a cRWA
and let $u_w:X_w\to U$ be a feedback-control policy induced by a CLF $w$ associated to $\RWA$ with basin of attraction $X_w$. Then, for all $x \in X_w$ and for all solutions $\sol_{x,u_w}$ of $\cS$, it holds that 
\begin{enumerate}[leftmargin=*]
 \item[(i)] $\sol_{x,u_w}(t)\in X_w$ for all $t\in\R_+$,\label{prop:result:high-to-low:item1}
	\item[(ii)] every trace $\trace \in \traces_{\labelfunc,\ldist}(\sol_{x,u_w})$ satisfies $\phi_{\Cc_w}$ in \eqref{eq:feas_control}, with $\Cc_w\in \AP_C$ being the control proposition associated to $w$ and $\Upsilon:=\kappa^\omega$ inducing a constant context.
\end{enumerate}
\end{proposition}

\begin{example}\label{example:CLFs}
Consider the robot example given in \cref{fig:example}, the cRWAs $\RWAa(d,e)$ and $\RWAe(d,e)$ as given in \cref{ex:RWAs}. A possible set of corresponding CLFs $w_a$ and $w_e$ with basins of attraction $X_a$ and $X_e$, respectively are depicted in \cref{fig:basins}.\qed
\end{example}

\begin{figure}
	\centering
	%% Creator: Inkscape 1.2.2 (b0a84865, 2022-12-01), www.inkscape.org
%% PDF/EPS/PS + LaTeX output extension by Johan Engelen, 2010
%% Accompanies image file 'basin.eps' (pdf, eps, ps)
%%
%% To include the image in your LaTeX document, write
%%   \input{<filename>.pdf_tex}
%%  instead of
%%   \includegraphics{<filename>.pdf}
%% To scale the image, write
%%   \def\svgwidth{<desired width>}
%%   \input{<filename>.pdf_tex}
%%  instead of
%%   \includegraphics[width=<desired width>]{<filename>.pdf}
%%
%% Images with a different path to the parent latex file can
%% be accessed with the `import' package (which may need to be
%% installed) using
%%   \usepackage{import}
%% in the preamble, and then including the image with
%%   \import{<path to file>}{<filename>.pdf_tex}
%% Alternatively, one can specify
%%   \graphicspath{{<path to file>/}}
%% 
%% For more information, please see info/svg-inkscape on CTAN:
%%   http://tug.ctan.org/tex-archive/info/svg-inkscape
%%
\begingroup%
  \makeatletter%
  \providecommand\color[2][]{%
    \errmessage{(Inkscape) Color is used for the text in Inkscape, but the package 'color.sty' is not loaded}%
    \renewcommand\color[2][]{}%
  }%
  \providecommand\transparent[1]{%
    \errmessage{(Inkscape) Transparency is used (non-zero) for the text in Inkscape, but the package 'transparent.sty' is not loaded}%
    \renewcommand\transparent[1]{}%
  }%
  \providecommand\rotatebox[2]{#2}%
  \newcommand*\fsize{\dimexpr\f@size pt\relax}%
  \newcommand*\lineheight[1]{\fontsize{\fsize}{#1\fsize}\selectfont}%
  \ifx\svgwidth\undefined%
    \setlength{\unitlength}{71.55414704bp}%
    \ifx\svgscale\undefined%
      \relax%
    \else%
      \setlength{\unitlength}{\unitlength * \real{\svgscale}}%
    \fi%
  \else%
    \setlength{\unitlength}{\svgwidth}%
  \fi%
  \global\let\svgwidth\undefined%
  \global\let\svgscale\undefined%
  \makeatother%
  \begin{picture}(1,1.23686948)%
    \lineheight{1}%
    \setlength\tabcolsep{0pt}%
    \put(0,0){\includegraphics[width=\unitlength]{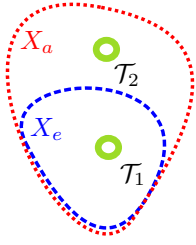}}%
    \put(0.60161939,0.29267923){\color[rgb]{0,0,0}\makebox(0,0)[lt]{\lineheight{1.25}\smash{\begin{tabular}[t]{l}$\Tc_1$\end{tabular}}}}%
    \put(0.57025247,0.81803353){\color[rgb]{0,0,0}\makebox(0,0)[lt]{\lineheight{1.25}\smash{\begin{tabular}[t]{l}$\Tc_2$\end{tabular}}}}%
    \put(0.12852899,0.53353364){\color[rgb]{0,0,1}\makebox(0,0)[lt]{\lineheight{1.25}\smash{\begin{tabular}[t]{l}$X_e$\end{tabular}}}}%
    \put(0.07612123,0.99023107){\color[rgb]{1,0,0}\makebox(0,0)[lt]{\lineheight{1.25}\smash{\begin{tabular}[t]{l}$X_a$\end{tabular}}}}%
  \end{picture}%
\endgroup%

	\caption{
	$X_a$ (region enclosed by red dotted line) and $X_e$ (region enclosed by blue dashed line) illustrate possible basins of attraction for the CLFs implementing the cRWAs $\RWAa(d,e)$ (ensuring to reach $\Tc_1$ while avoiding only the walls) and $\RWAe(d,e)$ (ensuring to reach $\Tc_1$ while avoiding walls and $\Tc_2$), respectively from \cref{ex:RWAs}.
	}\label{fig:basins}
\end{figure}

\subsection{The Bottom-Up Interface}\label{sec:algo:low-to-high}
%!TEX root = ../main.tex
The synthesis procedure from \cref{sec:algo:high-to-low} results in a finite set $\allCLF$ of CLFs with a finite set $\allU$ of control policies, such that each control policy $u_w\in \allU$ (resulting from a CLF $w\in\allCLF$) is equipped with a basin of attraction $X_w\subseteq X$, associated to a given $\RWA\in \allRWA$ resulting from a particular edge in the high-level synthesis game $\game^I$. This implies that whenever $w$ is non-global, i.e., if $X_w\subsetneq X$, the control policy $u_w$ cannot be applied anywhere.

Thinking back to the logical strategy computed in \cref{sec:algo:high}, policy $u_w$ must be used when its corresponding cRWA $\RWA$ for an edge $e$ is \enquote{activated} by a logical control strategy, \enquote{choosing} the edge $e$ in $\game^I$. By constructing the cRWA's for winning edges as defined in \cref{def:generatedRWA}, we essentially equip the resulting controller with a direct actuation capability of the underlying dynamical system -- it must choose between available feedback-control policies. To reflect this change of actuation capabilities in the higher-level game, we introduce a controller proposition $\Cc_w\in\AP_C$ for every available feedback-control policy $u_w$ which flags that $u_w\in \allU$ should be used to actuate $\cS$. Further, as every $u_w$ is equipped with a basin of attraction $X_w$, the resulting hybrid controller is implementable only  if the current continuous state $x$ is in $X_w$
We therefore need to track this information in the synthesis game. 
For this purpose, we introduce a new state proposition $\Xc_w$ for every $u_w\in \allU$ that flags whether the state is in its basin of attraction, and we define $\AP_S^+:=\AP_S\cup \bigcup_{w\in \allCLF}\Xc_w$ as the set of all state propositions including all additional state propositions $\Xc_w$'s.

The next four steps provide an algorithm that ensures that this information gets translated from the lower to the higher layer in a certified way (realizing the cyan marked transitions in \cref{fig:overview}), such that the resulting higher-layer synthesis game allows to synthesize a hybrid controller that solves \cref{prob:MainProb}.

\subsubsection{Changing Actuation Capabilities}\label{step4a}
As discussed before, in the initial game, the controller can activate/deactivate all state propositions in $\AP_S$.
However, in order to prepare the high-layer initial game $\game^I$ from \cref{sec:algo:high} for the incorporation of a refined system model, we need to incorporate the control propositions $\AP_C$ and make sure that these are the only propositions the controller can choose with its strategy, leading to the desired direct actuation of lower-level feedback control policies.
In particular, first, we need to ensure that all state propositions and observation propositions can only be activated/deactivated by the environment player.

This is achieved by updating the initial game to a merged game $\game^M$ (i.e., going from \circled{2} to \circled{6} in \cref{fig:overview}) while preserving the parity condition and a one-to-one correspondence between the traces generated by plays in $\game^I$ and the ones generated by plays in $\game^M$.

\begin{definition}\label{def:mergedgame}
Given an initial game $\game^I = (\gamegraph^I,\paritygame(\priority^I))$ with game graph $\gamegraph^I=(\vertices^I,\edges^I,\labelfuncG^I)$, the \emph{merged game} $\game^M = (\gamegraph^M,\paritygame(\priority^M))$ with game graph $\gamegraph^M=(\vertices^M,\edges^M,\labelfuncG^M)$ is constructed as follows.
\begin{itemize}[leftmargin=*]
    \item The set of $\p{1}$ vertices is preserved, i.e., $\verticesl^M = \verticesl^I$ s.t.\ for each $v\in\verticesl^M$, $\priority^M(v) = \priority^I(v)$ and $\labelfuncG^M(v) = \emptyset$.
  \item For every pair of $\p{1}$ vertices $v_1,v_2\in\verticesl^I$ connected via a $\p{0}$ vertex $v_0\in\verticeso^I$, i.e.,$(v_1,v_0),(v_0,v_2)\in \edges^I$, we add:
    \begin{itemize}
        \item a unique $\p{0}$ vertex $v\in\verticeso^M\setminus\verticeso^I$ with $\labelfuncG^M(v) = \labelfuncG^I(v_0)\cup\labelfuncG^I(v_2)$ and $\priority^M(v) = \priority^I(v_0)$, 
        \item new edges $(v_1,v),(v,v_2)\in\edges^M\setminus\edges^I$.
    \end{itemize}
\end{itemize}
\end{definition}

This leads to the following lemma. 

\begin{lemma}\label{lemma:mergedgame}
Let $\game^I$ be the parity game constructed from $\spec$ over $\AP$ as in \cref{prop:result:high} and $\game^M$ its merged version constructed via \cref{def:mergedgame}. Then $\game^M$ is total w.r.t. $\AP$, and every winning play in $\game^M$ generates a trace which satisfies~$\spec$.
\end{lemma}

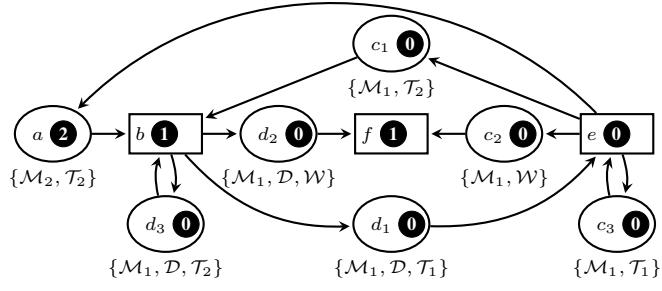
\begin{figure*}
	\scriptsize
	\centering
	\begin{tikzpicture}
		\node[bplayer0,label={[align=center]below:$\{\Mc_2,\Tc_2\}$}] (a) at (0, 0) {$a$\bcircled{2}};
		\node[bplayer1] (b) at (\hpos,0) {$b$\bcircled{1}};
	
		\node[bplayer0,label={[align=center]below:$\{\Mc_1,\Tc_2\}$}] (c1) at (3*\hpos, \bypos) {$c_1$\bcircled{0}};
		\node[bplayer0,label={[align=center]below:$\{\Mc_1,\Wc\}$}] (c2) at (4*\hpos, 0) {$c_2$\bcircled{0}};
		\node[bplayer0,label={[align=center]below:$\{\Mc_1,\Tc_1\}$}] (c3) at (5*\hpos, -\bypos) {$c_3$\bcircled{0}};
		
		\node[bplayer0,label={[align=center]below:$\{\Mc_1,\Dc,\Tc_1\}$}] (d1) at (3*\hpos, -\bypos) {$d_1$\bcircled{0}};
		\node[bplayer0,label={[align=center]below:$\{\Mc_1,\Dc,\Wc\}$}] (d2) at (2*\hpos, 0) {$d_2$\bcircled{0}};
		\node[bplayer0,label={[align=center]below:$\{\Mc_1,\Dc,\Tc_2\}$}] (d3) at (1*\hpos, -\bypos) {$d_3$\bcircled{0}};
	
		\node[bplayer1] (e) at (5*\hpos, 0) {$e$\bcircled{0}};
		\node[bplayer1] (f) at (3*\hpos, 0) {$f$\bcircled{1}};
	
		\path[->] (a) edge (b);
		\path[->] (b) edge[bend left =15] (d3) edge[bend left=-25] (d1) edge (d2);
		\path[->] (c1) edge (b);
		\path[->] (c2) edge (f);
		\path[->] (c3) edge[bend left=15] (e);
		\path[->] (d1) edge[bend left=-25] (e);
		\path[->] (d2) edge (f);
		\path[->] (d3) edge[bend left=15] (b);
		\path[->] (e) edge[bend left =15] (c3) edge (c1) edge (c2) edge[bend left=-45] (a);
	\end{tikzpicture}
	\caption{Corresponding merged game for the initial game given in \cref{fig:gamegraph}, where labels of $\p{1}$ vertices are empty sets.}\label{fig:mergedgame}
	\end{figure*}

\begin{proof}
Let $\play = v_0v_1\cdots$ be a winning play in $\game^M$ with $v_{2k}\in V^M_1$ for every $k\geq 0$, and let $\trace = l_0l_1\cdots$ be the trace generated by the play $\play$. Then by construction, vertices $v_{2k}$ also belong to $V^I_1$ with same priority, i.e., $\priority^M(v_{2k}) = \priority^I(v_{2k})$ for every $k\geq 0$. Furthermore, for every $v_{2k+1}\in V^M_0$, there exists a corresponding vertex $v_{2k+1}'\in V^I_0$ that connects the vertices $v_{2k}$ and $v_{2k+2}$ in the game $\game^I$ such that $\priority^M(v_{2k+1}) = \priority^I(v_{2k+1}')$ and $\labelfuncG^M(v_{2k+1}) = \labelfuncG^I(v_{2k+1})\cup\labelfuncG^I(v_{2k+2})$.
Hence, the play $\play' = v_0v_1'v_2\cdots$ is a winning play in game $\game^I$ as maximum priority seen infinitely often in $\play'$ w.r.t. $\priority^I$ is same as the maximum priority seen infinitely often in $\play$ w.r.t $\priority^M$.
Now, let $\trace' = l_0'l_1'\cdots$ be the trace generated by $\play'$ in $\game^I$, then by construction of game $\game^I$, $\trace'$ satisfies the specification $\spec$.
Moreover, since $\labelfuncG^M(v_{2k+2}) = \emptyset$ for every $k\geq 0$, we have, by definition, $l_k = \labelfuncG^M(v_{2k+1})\cup \labelfuncG^M(v_{2k+2}) = \labelfuncG^M(v_{2k+1})$. Therefore, $l_k'= \labelfuncG^I(v_{2k+1})\cup\labelfuncG^I(v_{2k+2}) = l_k$. So, $\trace = \trace'$, and hence, $\trace$ satisfies the specification $\spec$.

Using similar arguments, it can be shown that for every play in $\game^I$, there exists a corresponding play in $\game^M$ that generates the same trace.
Hence, as $\game^I$ is total w.r.t. $\AP$, so is $\game^M$.
\end{proof}

\begin{example}
Consider the initial game $\game^I$ given in \cref{fig:gamegraph}. Then the resulting merged game $\game^M$ is depicted in \cref{fig:mergedgame}.
As shown in the figure, $\p{1}$ vertices, i.e., vertices $b,e,f$, are preserved with same priorities but empty labels.
For every pair of $\p{1}$ vertices connected via a $\p{0}$ vertex in $\game^I$, there is a new vertex with label containing all necessary propositions that connects the pair in $\game^M$, e.g., for vertex $b$ and $f$ connected via $d$ in $\game^I$, the new vertex $d_2$ containing labels of both $d$ and $f$ connects vertex $b$ and $f$ .\qed
\end{example}

Note that we still have not explicitly incorporated the control propositions in the merged game. In the next steps, we will introduce the control propositions that are realizable by low-level feedback control and incorporate them into the high-level game graph.

\subsubsection{Control Graph Construction}\label{step4b}
In this step we construct a game graph that captures the interplay of the environment and observation propositions contained in the context $\contextRWA$ of a given cRWA (i.e., going from \circled{5} to \circled{7} in \cref{fig:overview}) with the newly introduced control and state propositions $\Cc_w\in \AP_C$ and $\Xc_w\in \AP_S^+$.
Intuitively, this graph captures which context changes an application of a particular feedback control policy $u_w$ for a CLF $w$ (triggered by $\Cc_w$) might cause. When composed with the modifided game graph $G^M$ from \cref{sec:algo:low-to-high} this leads to the \emph{lazy} refinement of the logical synthesis game discussed earlier, which only includes relevant information about the low-level feedback control loop. 

Let us denote the cRWA's for which the CLF $w$ was synthesized by $\RWA_w = (\contextRWA_w,\reachRWA_w,\avoidRWA_w)$. 
Consider $\AP_S^+\supseteq\AP_S$ the set of all state propositions including all additional state propositions $\Xc_w$'s as defined above, and $\labelfunc^+\colon X \rightarrow 2^{\AP_S^+}$ be an extended version of labelling function $\labelfunc$ defined by $\labelfunc^+(x)=\{\Xc\in \AP^+_S\;\vert\;x\in\Xc\}$,  (and thus, $\labelfunc^+(x)\cap\AP_S = \labelfunc(x)$ for all $x\in X$).

\begin{definition}\label{def:controlgraph}
	Given the control system  $\cS:=(X,U,f)$ with labelling function $\labelfunc^+$ and the set $\allCLF$ of all CLFs computed as before,
	the \emph{control game graph} $\gamegraph^C=(V^C,E^C,\labelfuncG^C)$ with $\labelfuncG^C\colon \vertices\rightarrow2^{\AP_S^+\cup\AP_O}$ is defined as follows.
	\begin{enumerate}[leftmargin=*]
		\item\label{item:def:controlgraph:1} For each CLF $w\in\allCLF$, there are two $\p{1}$ vertices in $\verticesl^C$, a \emph{transition} vertex and an \emph{invariant} vertex, both with label $\{\Cc_w\}$.
		\item\label{item:def:controlgraph:2} For every subset of propositions $c \subseteq {\AP_O\cup \AP_S^+}$, there is a Player 0 vertex $v\in \verticeso^C$ with $\labelfuncG^C(v)=c$ iff there exists $x\in X$ such that $c\cap \AP_S^+ = \labelfunc^+(x)$.
		\item\label{item:def:controlgraph:3} From each invariant vertex $v\in \verticesl^C$ of some CLF $w$, there is an edge $(v,v')$ to $v'\in \verticeso^C$ 
  		iff $\reachRWA_w\subseteq \labelfuncG^C(v')$.
		\item\label{item:def:controlgraph:4} From each transition vertex $v\in \verticesl^C$ of some CLF $w$, there is an edge $(v,v')$ to $v'\in \verticeso^C$  iff $\Xc_w\in\labelfuncG^C(v')$.
		\item\label{item:def:controlgraph:5} From each $\p{0}$ vertex $v\in \verticeso^C$ with $\Xc_w \in \labelfuncG^C(v)$ and $\contextRWA_w = \labelfuncG^C(v)\cap\AP_O$ for some CLF $w$, 
		if $\reachRWA_w \subseteq \labelfuncG^C(v)$, then 
		there is an edge to the invariant vertex of $w$, else, then there is an edge to the transition vertex of~$w$.
		\end{enumerate}
\end{definition}

The construction of $\gamegraph^C$ via \cref{def:controlgraph} translates some characteristics of the low-level continuous closed loop system captured by  \cref{prop:result:high-to-low} into the higher-layer synthesis game.
In addition, it ensures that a logical controller actuating a control policy $u_w$ via control proposition $\Cc_w$ can only do so if context $\contextRWA_w$ is true and the continuous system is in the basin of attraction $X_w$ (signaled by the system proposition $\Xc_w$ being true). 
These translations can be formalized via LTL formulas which are ensured to hold true on every play over $\gamegraph^C$ as formalized in the next lemma.

\begin{lemma}\label{lemma:controlgraph}
Given the premises of \cref{def:controlgraph}, it holds for every trace $\trace$ over $\gamegraph^C$ and every CLF $w\in\allCLF$ with basin of attraction $\Xc_w$, cRWA $\RWA_w := (\contextRWA_w,\reachRWA_w,\avoidRWA_w)$ and associated controller $\Cc_w$, that 
\begin{eqnarray}
	&\globally (\Xc_w \Rightarrow \neg\avoidRWA_w),\label{eq:lemma:controlgraph1}\\
	&\globally (\Cc_w \Rightarrow \Xc_w \wedge \contextRWA_w),\label{eq:lemma:controlgraph2}\\
	&\globally (\reachRWA_w \wedge \Cc_w \Rightarrow \Next \reachRWA_w).\label{eq:lemma:controlgraph3}\\
	&\globally (\Xc_w \wedge \Cc_w \Rightarrow \Next \Xc_w).\label{eq:lemma:controlgraph4}
\end{eqnarray}
\end{lemma}
\begin{proof}
Let $\play=v_0v_1\cdots$ be a play in $\gamegraph^C$ and $\trace=l_0l_1\cdots$ be the trace generated by $\play$. We need to show that $\trace$ satisfies \eqref{eq:lemma:controlgraph1}-\eqref{eq:lemma:controlgraph4}.
By \cref{def:RWAtoPolicy}, for each $w\in\allCLF$, $X_w\cap\avoidRWA_w=\emptyset$. Then, by item~\ref{item:def:controlgraph:2}, for each $i\geq 0$, if $\Xc_w\in\labelfuncG^C(v_i)$ then $\avoidRWA_w\cap\labelfuncG^C(v_i)=\emptyset$. 
Hence, $\trace$ satisfies \eqref{eq:lemma:controlgraph1}.
Next, by item~\ref{item:def:controlgraph:5}, if $\Cc_w\in\labelfuncG^C(v_{i+1})$ for some $i\geq 0$, then $\Xc_w\in\labelfuncG^C(v_i)$ and $\contextRWA_w = \labelfuncG^C(v_i)\cap\AP_O$. Hence, $\trace$ also satisfies~\eqref{eq:lemma:controlgraph2}.
Next, by item~\ref{item:def:controlgraph:3} and item~\ref{item:def:controlgraph:5}, if $\reachRWA_w\subseteq\labelfuncG^C(v_i)$ and $\Cc_w\in\labelfuncG^C(v_{i+1})$ for some $i\geq 0$, then $\reachRWA_w\subseteq\labelfuncG^C(v_{i+2})$. 
Hence, $\trace$ also satisfies \eqref{eq:lemma:controlgraph3}.
Similarly, by item~\ref{item:def:controlgraph:4} and \ref{item:def:controlgraph:5}, if $\Xc_w\in\labelfuncG^C(v_i)$ and $\Cc_w\in\labelfuncG^C(v_{i+1})$ for some $i\geq 0$, then $\Xc_w\in\labelfuncG^C(v_{i+2})$. 
Hence, $\trace$ satisfies \eqref{eq:lemma:controlgraph4}.
\end{proof}

Intuitively, given the premises of \cref{lemma:controlgraph}, equations \eqref{eq:lemma:controlgraph1}-\eqref{eq:lemma:controlgraph4} ensures the following low-level properties on the game graph level.
First, 
\eqref{eq:lemma:controlgraph1} ensures that the basin of attaction $\Xc_w$ does not have an intersection with the avoid region $\avoidRWA_w$.
Next, \eqref{eq:lemma:controlgraph2} ensures that the controller $\Cc_w$ can only be applied if the system is within the corresponding basin of attaction $\Xc_w$ and the context $\kappa_w$ holds. Note that this does not restrict the environment from changing the context right after the feedback control policy associated with $\Cc_w$ was applied.
Finally, \eqref{eq:lemma:controlgraph3}-\eqref{eq:lemma:controlgraph4} ensures that if the system is within the target region $\reachRWA_w$ (resp. the basin of attaction $\Xc_w$) and the controller $\Cc_w$ is applied, the system cannot leave $\reachRWA_w$ (resp.\ $\Xc_w$).

In total, the control game graph $G^C$ models all the state proposition sequences generated by a trajectory $\xi$ triggered by the controller policies associated with $\allCLF$ as in \cref{prop:result:high-to-low}.
Furthermore, it also models the logical disturbances received as inputs via the disturbance function $\ldist\in\Ldist$. This is formalized by the next lemma which directly follows by item~\ref{item:def:controlgraph:2}-\ref{item:def:controlgraph:4} of \cref{def:controlgraph}.
\begin{lemma}\label{lemma:controlgraphtotal}
Given the premises of \cref{def:controlgraph}, 
one of the following holds for every disturbance function $\ldist\in\Ldist$
\begin{itemize}
	\item for some play in $\gamegraph^C$, its generated trace $\trace$ satisfies that $\trace|_{\AP_O}=\ldist$, or
 	\item for some play in $\gamegraph^C$ ending in a $\p{0}$ dead-end, its generated trace $\trace$ satisfies that $\trace|_{\AP_O}$ is a prefix of $\ldist$.
\end{itemize}

\end{lemma}

\begin{example}\label{example:controlgraph}
For the CLFs $w_a$ and $w_e$ given in \cref{example:CLFs} with basins of attraction $X_a$ and $X_e$ as shown in \cref{fig:basins},
the corresponding control game graph (without the $\p{0}$ dead-ends) is depicted in \cref{fig:controlgraph}.
As in the figure, the transition vertices of $w_a$ and $w_e$ are vertices $a$ and $c$, respectively, and the invariant vertices are vertices $g$ and $i$, respectively. 
Note that both CLFs have context $\{\Mc_1,\Dc\}$. Hence, 
vertices with a label that contains $\Xc_a$ or $\Xc_e$ but not the propositions $\Mc_1$ or $\Dc$ are $\p{0}$ dead-ends (no outgoing edges are defined from them).
For simplicity, those vertices are not shown in \cref{fig:controlgraph}.\qed
\end{example}
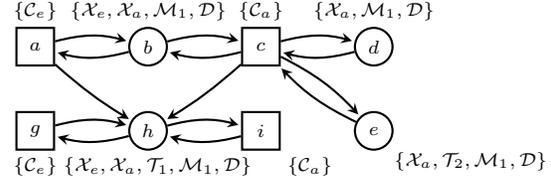
\begin{figure}
	\centering
	\scriptsize
	\begin{tikzpicture}
		\node[player1,label={[align=center]above:$\{\Cc_e\}$}] (a) at (0, 0) {$a$};
		\node[player0,label={[align=center]above:$\{\Xc_e,\Xc_a,\Mc_1,\Dc\}$}] (b) at (\hpos,0) {$b$};
		\node[player1,label={[align=center]above:$\{\Cc_a\}$}] (c) at (2*\hpos, 0) {$c$};
		\node[player0,label={[align=center]above:$\{\Xc_a,\Mc_1,\Dc\}$}] (d) at (3*\hpos, 0) {$d$};
		\node[player0,label={[align=center]below right:$\{\Xc_a,\Tc_2,\Mc_1,\Dc\}$}] (e) at (3*\hpos, -0.7*\ypos) {$e$};
		
		\node[player1,label={[align=center]below:$\{\Cc_e\}$}] (g) at (0*\hpos, -0.7*\ypos) {$g$};
		\node[player0,label={[align=center]below:$\hspace*{1em}\{\Xc_e,\Xc_a,\Tc_1,\Mc_1,\Dc\}$}] (h) at (1*\hpos, -0.7*\ypos) {$h$};
		\node[player1,label={[align=center]below right:$\{\Cc_a\}$}] (i) at (2*\hpos, -0.7*\ypos) {$i$};

		\path[->] (a) edge[bend left=15] (b) edge[bend left=-5] (h);
		\path[->] (b) edge[bend left=15] (a) edge[bend left=15] (c);
		\path[->] (c) edge[bend left=15] (b) edge[bend left=15] (d) edge[bend left=5] (h) edge[bend left=10] (e);
		\path[->] (d) edge[bend left=15] (c);
		\path[->] (e) edge[bend left=10] (c);
		
		\path[->] (g) edge[bend left=15] (h);
		\path[->] (h) edge[bend left=15] (g) edge[bend left=15] (i);
		\path[->] (i) edge[bend left=15] (h);
	\end{tikzpicture}
	\caption{The corresponding control game graph (without $\p{0}$ dead-ends) for the basins of attraction in \cref{fig:basins}.}\label{fig:controlgraph}
	\end{figure}

While we could now take the product of $\gamegraph^C$ with $\game^M$ from the previous step in order to obtain the new, refined logical synthesis game, we note that this typically does not lead to a game that actually has a winning strategy. The reason for this lies in the fact that the modification of $\game^I$ to $\game^M$ gives the right to trigger state propositions to the environment, i.e., now the controller actuates $\AP_C$ and gets \enquote{notified} by the underlying dynamical systems via a triggering of $\AP_S$'s that the actuated controller actually resulted in the (hopefully desired) state proposition change. 
From a two-player game perspective, the environment could now use its additional power to prevent the robot to reach the target. E.g., in \cref{fig:controlgraph}, starting from vertex~$b$, if the controller keeps using the control policy for CLF $w_e$, then the environment can force the play to loop between vertex $a$ and $b$ instead of reaching target $\Tc_1$ represented by vertex~$h$.
This is because the resulting logical game still misses essential information about the low-level closed loop dynamics under a given feedback-control policy. We thus incorporate, in what follows, the information captured by item (ii) of \cref{prop:result:high-to-low}.

\subsubsection{Persistent Live-Groups}\label{step5}
In order to capture item (ii) of \cref{prop:result:high-to-low} in the logical synthesis game,  we construct so called \emph{persistent liveness constraints} (i.e., going from \circled{5} to \circled{8} in \cref{fig:overview}) to annotate the control game graph $\game^C$ which are inspired by progress groups from \cite{progress_groups}.

\begin{definition}\label{def:perslivegroups}
Given a game graph $\gamegraph = (\vertices,\edges)$, a \emph{persistent live-group} is a tuple $(\perssource,\persedges,\perstarget)$ consisting of sets $\perssource,\perstarget\subseteq \vertices$ and $\persedges\subseteq \edgeso$ such that $\perstarget\subseteq \perssource$.
The constraints represented by such a persistent live group is expressed by the following LTL formula
\begin{equation}\label{eq:assumpPers}
	\assumpPers(\perssource,\persedges, \perstarget) \coloneqq \globally \big(\globally (\perssource\wedge \assumpC(\persedges)) \Rightarrow \finally \perstarget\big),  
\end{equation}
where $\assumpC(\persedges) \coloneqq \src(\persedges) \Rightarrow \persedges$.
Moreover, the constraints represented by a set $\perslivegroups$ of persistent live-groups is denoted by $\assumpPers(\perslivegroups) \coloneqq \bigwedge_{(\perssource,\persedges,\perstarget)\in\perslivegroups} \assumpPers(\perssource,\persedges, \perstarget)$.
\end{definition}

Intuitively, $\assumpC(\persedges)$ ensures that edges in $\persedges$ are chosen when possible, as this is only possible for $\p{0}$ vertices in $\perssource$.
Furthermore, \eqref{eq:assumpPers} ensures that persistently choosing the edges in $\persedges$ from the source vertices $\perssource$ will eventually lead us to a vertex in $\perstarget$.

For a CLF $w\in\allCLF$, we construct a persistent live-group $(\perssource_w,\persedges_w, \perstarget_w)$ that captures \cref{prop:result:high-to-low} in the following way. Given the control graph $\gamegraph^C$ as defined before, and a CLF $w\in \allCLF$, first, the  persistent activation of $\Cc_w$ is captured via the set $\persedges_w$ collecting all ($\p{0}$) edges that end in vertices with labeled by $\Cc_w$, i.e.,
\begin{equation}\label{eq:persedges}
	\persedges_w = \edges \cap \big(\vertices\times \{\vertex\in \vertices \mid \Cc_w\in\labelfuncG^C(\vertex)\}\big).
\end{equation}
Always choosing an edge from $\persedges_w$ will force $\Xc_w$ to remain true within the same context $\contextRWA_w$, which is captured by the set $\perssource_w$ collecting all ($\p{0}$) vertices labeled by $\Xc_w$ and propositions in $\contextRWA_w$, and all ($\p{1}$) vertices labeled by $\Cc_w$, i.e., 
\begin{multline}\label{eq:perssource}
	\perssource_w = \{\vertex\in \vertices \mid \Xc_w\in\labelfuncG^C(\vertex), \contextRWA_w = \labelfuncG^C(\vertex)\cap\AP_O\}\\
	\cup \{\vertex\in \vertices \mid \Cc_w\in\labelfuncG^C(\vertex)\}.
\end{multline}
Finally, we know that always choosing an edge from $\persedges_w$ will eventually lead us to a vertex where $\reachRWA_w$ is true, captured by the set $\perstarget_w$ collecting all vertices labeled by $\reachRWA_w$, i.e., 
\begin{equation}\label{eq:perstarget}
	\perstarget_w = \{\vertex\in\vertices \mid  \reachRWA_w = \labelfuncG^C(\vertex) \cap \AP_S \}.
\end{equation}

\begin{example}
For example, consider the control game graph shown in \cref{fig:controlgraph} for \cref{example:controlgraph}.
For CLF $w_1$ of $\RWA_1$, the corresponding persistent live-group is $(\perssource,\persedges,\perstarget)$, where $\perssource = \{a,b,g,h\}$ corresponds to the region of basin of attraction for $w_1$ with context $\contextRWA_1=\{\Mc_1,\Dc\}$ being true, $\persedges=\{e_{ba},e_{hg}\}$ corresponds to the edges that represent using the control policy $u_w$, and $\perstarget=\{h\}$ corresponds to the target region of $\RWA_1$, i.e., vertices labeled by $\Tc_1$ .\qed
\end{example}

Given the set $\allCLF$ of all CLFs as given before, we collect all the corresponding persistent live-groups for the CLFs in $\allCLF$ in the set $\perslivegroups^C$. 
With the persistent live-group assumptions $\perslivegroups^C$, the control game graph $\gamegraph^C$ also ensures that item (ii) of \cref{prop:result:high-to-low} holds at a higher level as formalized below.
\begin{lemma}\label{lemma:augmentedgamegraph}
Let $\gamegraph^C$ be a control graph as in \cref{def:controlgraph} and $\allCLF$ a set of CLFs with persistent live-groups $(\perssource_w,\persedges_w,\perstarget_w)$ for all $w\in\allCLF$ as in \eqref{eq:perssource}-\eqref{eq:perstarget}. Then a play over $\gamegraph^C$ satisfies $\assumpPers(\perssource_w,\persedges_w,\perstarget_w)$, if and only if its generated trace satisfies
\begin{equation}\label{eq:augmentedgamegraph1}
	\globally (\globally (\Xc_w \wedge \contextRWA_w \wedge \Cc_w) \Rightarrow \finally \reachRWA_w).
\end{equation}
Moreover, \eqref{eq:augmentedgamegraph1} along with \eqref{eq:lemma:controlgraph1}-\eqref{eq:lemma:controlgraph4} ensures that every trace generated by plays in $\gamegraph^C$ satisfying $\assumpPers(\perssource_w,\persedges_w,\perstarget_w)$ also satisfies $\spec_{\Cc_w}$ in \eqref{eq:feas_control}.
Conversely, every trace satisfying \eqref{eq:feas_control} is generated by a play in $\gamegraph^C$ satisfying $\assumpPers(\perssource_w,\persedges_w,\perstarget_w)$.
\end{lemma}

\begin{proof}
Let $\play=v_0v_1\cdots$ be a play in $\gamegraph^C$ and $\trace=l_0l_1\cdots$ be the trace generated by $\play$.
By the definition of the persistent live-groups as in \eqref{eq:persedges}-\eqref{eq:perstarget}, rewriting \eqref{eq:assumpPers} in terms of propositions gives us that, $\play$ satisfies $\assumpPers(\perssource_w,\persedges_w,\perstarget_w)$ if and only if trace $\trace$ satisfies \eqref{eq:augmentedgamegraph1}.
Furthermore, by \cref{lemma:controlgraph}, the trace $\trace$ also satisfies \eqref{eq:lemma:controlgraph1}-\eqref{eq:lemma:controlgraph4}.

Now, suppose $\trace$ satisfied \eqref{eq:augmentedgamegraph1}, then we need to show that $\trace$ also satisfies $\spec_{\Cc_w}$ in \eqref{eq:feas_control}. It suffices to show that for every $k\geq 0$, the trace $\trace_k = l_kl_{k+1}\cdots$ satisfies the following:
\[\square (\Cc_w ~\wedge~\contextRWA_w) \Rightarrow \Diamond\square \reachRWA_w ~ \wedge ~  \square \neg\avoidRWA_w.\]
Suppose $\trace_k$ satisfies $\square (\Cc_w\wedge\contextRWA_w)$. 
Then, every $j\geq k$, $l_j$ satisfies $\Cc_w$, which implies, by \eqref{eq:lemma:controlgraph2}, $l_j$ also satisfies $\Xc_w$.
Moreover, by \eqref{eq:lemma:controlgraph1}, $l_j$ also satisfies $\neg\avoidRWA_w$ for each $j\geq 0$.
Therefore, trace $\trace_k$ satisfies both $\square (\Cc_w \wedge\contextRWA_w\wedge \Xc_w)$ and $\square\neg\avoidRWA_w$, which then implies, by \eqref{eq:augmentedgamegraph1}, $\trace_k$ also satisfies $\finally\reachRWA_w$.
That means, there exists $m\geq k$ such that $l_m$ satisfies $\reachRWA_w$.
As $l_m$ also satisfies $\Cc_w$, by \eqref{eq:lemma:controlgraph3}, $l_{m+1}$ satisfies $\reachRWA_w$. 
Using the same argument inductively, we can show that $l_i$ satisfies $\reachRWA_w$ for all $i\geq m$.
Therefore, $\trace_k$ satisfies both $\finally\globally\reachRWA_w$ and $\globally\neg\avoidRWA_w$.
Conversely, suppose $\trace$ satisfies \eqref{eq:feas_control}, then we need to show that $\play$ satisfies $\assumpPers(\perssource_w,\persedges_w,\perstarget_w)$.
It is enough to show that $\trace$ satisfies \eqref{eq:augmentedgamegraph1}, which trivially follows from \eqref{eq:feas_control}.
\end{proof}

\subsubsection{Final Augmented Parity Game}\label{step6}
Given the three ingredients from the last steps, we are now ready to construct the final augmented (parity) game (i.e., going from \circled{6},\circled{7},\circled{8} to \circled{9} in \cref{fig:overview}) which serves a new logical synthesis game for the final hybrid controller and is defined next.
\begin{definition}\label{def:augmentedGame}
An \emph{augmented game} $\game$ is a tuple $(\gamegraph, \spec, \perslivegroups)$ consisting of a game graph $\gamegraph$, a set of persistent live-groups $\perslivegroups$ over $\gamegraph$ and an LTL specification $\spec$. Moreover, an augmented game $(\gamegraph, \spec, \perslivegroups)$ is equivalent to the game $(\gamegraph, \assumpPers(\perslivegroups)\Rightarrow \spec)$.
\end{definition}

Let us now describe how the final augmented parity game, i.e., an augmented game with parity specification, is constructed. Recall that $\verticesi^M$ and $\verticesi^C$ are the vertices of $\p{i}$ in game graph $\gamegraph^M$ and $\gamegraph^C$, respectively.
\begin{definition}\label{def:finalgame}
Given the merged game $\game^M$, control game graph $\gamegraph^C$, and persistent live-groups $\perslivegroups^C$ as computed before, the \emph{final augmented parity game} $\game^F = (\gamegraph^F,\paritygame(\priority^F),\perslivegroups^F)$ with $\gamegraph^F = (\vertices^F,\edges^F,\labelfuncG^F)$ is constructed by taking the product of the game $\game^M$ and the tuple $(\gamegraph^C,\perslivegroups^C)$ as follows:
\begin{itemize}[leftmargin=*]
\item $\vertex = (\vertex^M,\vertex^C)\in \verticesi^F$ with label $\labelfuncG^F(\vertex) = \labelfuncG'(\vertex^M) \cup \labelfuncG^C(\vertex^C)$ if $\vertex^M\in\verticesi^M$, $\vertex^C\in\verticesi^C$, and $\labelfuncG'(\vertex^M)|_{\AP_O\cup\AP_S} = \labelfuncG^C(\vertex^C)|_{\AP_O\cup\AP_S}$;
\item there is an edge $(\vertex_1,\vertex_2)\in \edges^F$ from $\vertex_1 = (\vertex^M_1,\vertex^C_1)$ to $\vertex_2 = (\vertex^M_2,\vertex^C_2)$ if $(\vertex^M_1,\vertex^M_2)\in \edges^M$ and $(\vertex^C_1,\vertex^C_2)\in E^C$;
\item for vertex $\vertex = (\vertex^M,\vertex^C)\in\vertices^F$, $\priority(\vertex) = \priority^M(\vertex^M)$;
\item $(\perssource,\persedges,\perstarget) \in \perslivegroups^F$ if there exists a $(\perssource^C,\persedges^C,\perstarget^C) \in \perslivegroups^C$ such that:
\begin{itemize}
    \item $\perssource = \vertices^F \cap (\vertices^M\times\perssource^C)$,
    \item $\perstarget = \vertices^F \cap (\vertices^M\times\perstarget^C)$,
    \item for every edge $e = (\vertex_1,\vertex_2)\in\edges^F$ with $\vertex_1 = (\vertex^M_1,\vertex^C_1)$ and $\vertex_2 = (\vertex^M_2,\vertex^C_2)$, it holds $e\in \persedges$ if and only if $(\vertex^C_1,\vertex^C_2)\in \persedges^C$.
\end{itemize}
\end{itemize}
\end{definition}

As the priority function $\priority^F$ is defined by the priority function $\priority^M$ of the merged game $\game^M$ and every winning play in $\game^F$ satisfying $\assumpPers(\perslivegroups^F)$ needs to satisfy the parity condition $\paritygame(\priority^F)$, the next proposition directly follows from \cref{lemma:mergedgame}.
\begin{proposition}\label{prop:low-high-winnnig}
Given the LTL specification $\spec$, initial game $\game^I$, and the final game $\game^F$ with persistent live-groups $\perslivegroups^F$ as in \cref{def:finalgame}, 
suppose $\trace$ be a trace generated by a winning play satisfying $\assumpPers(\perslivegroups^F)$ in $\game^F$, then 
$\trace$ satisfies the specification $\spec$.
\end{proposition}

\subsection{Solving the Final Augmented Game}\label{subsec:solveaugmentedgame}
%!TEX root = ../main.tex
As discussed in Section \ref{sec:algo:high}, the initial game $\game^I$ allowed the system to instantaneously activate or deactivate all state propositions in $\AP_S$. However, this was no longer possible in the merged game $\game^M$.
But, in the final game $\game^F$, the persistent live-groups, using the results described in \cref{lemma:augmentedgamegraph}, enable the system to activate or deactivate specific state propositions which are ensured to become \emph{eventually} true (using the associated feedback-control policy) if no external context change is induced.

The next obvious step of our synthesis procedure is to solve the final augmented game $\game^F$, i.e., to compute a winning strategy in this game (realizing the violet marked transitions in \cref{fig:overview}, i.e., going from \circled{9} to \circled{10}). Based on the observation made in \cref{def:augmentedGame} that an augmented game $(\gamegraph, \spec, \perslivegroups)$ is equivalent to the game $(\gamegraph, \assumpPers(\perslivegroups)\Rightarrow \spec)$ one can use standard game solving techniques for this purpose. This, however, usually results in computationally intractable problems. We will therefore provide a new algorithm for solving augmented parity games, in the subsequent \cref{subsec:DiscreteLayer}, which has a similar algorithmic structure and therefore also similar worst-case time  complexity as the standard algorithm for solving classical (non-augmented parity) games and therefore allows for a computationally tractable solution.

For the time being, we assume that we have solved $\game^F$, i.e., we have computed a winning region $\win^F\subseteq V^F$ and a winning strategy $\strat^F\colon \verticeso^F\rightarrow\verticesl^F$ s.t.\ all resulting traces satisfy $\spec$ due to \cref{prop:low-high-winnnig}.

\subsection{Constructing the Hybrid Controller}\label{subsec:hybridcontroller}
%!TEX root = ../main.tex
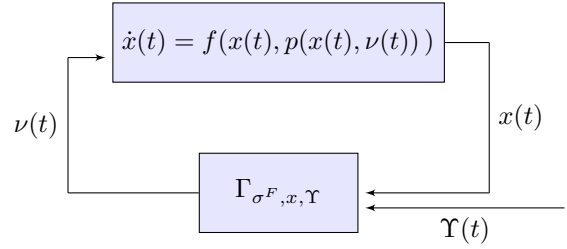
\begin{figure}[t!]
\tikzstyle{block} = [draw, fill=blue!10, rectangle, 
    minimum height=3em, minimum width=6em]
\tikzstyle{input} = [coordinate]
\tikzstyle{output} = [coordinate]
\tikzstyle{pinstyle} = [pin edge={to-,thin,black}]

\begin{tikzpicture}[auto, node distance=2cm,>=latex']
    \node [input, name=input1] {};
    \node [input, name=MID11 ,right of=input1,node distance=2.1cm] {};
    \node [block, right of=MID11, node distance=2.8cm] (system1) {$\dot x(t)=f(x(t), p(x(t),\nu(t))\,)$};
    \node [input, name=MID12 ,right of=system1,node distance=2.8cm] {};
    \node [output, right of=MID12,node distance=1cm ,name=output1]{};
    \node [block, below of=system1] (system2) {$\Gamma_{\sigma^F,x, \ldist}$};
    \node [input, below of=output1, node distance=2.2cm] (input2) {};
    \node [input, name=MID21 ,below of=MID12] {};
    \node [input, name=MID22 ,below of=MID11] {};
    \node [input, name=MID213 , below of=MID11,node distance=0.2cm] {};
    \node [input, name=MID123 , below of=MID12,node distance=2cm] {};
    \node [input, name=system11, below of=system1,node distance=0.2cm] {};
    \node [input, name=system22, below of=system2,node distance=0.2cm] {};

    \draw [-] (MID12) -- node[name=passaggio1] {$x(t)$} (MID123);
    \draw [->] (MID123) -- node[name=passaggio1] {} (6.05,-2);
    \draw [->] (input2) -- node[name=passaggio1] {$\ldist(t)$} (6.05,-2.2);
    \node [output, below of=input1] (output2) {};
    \node [input, name=MID22 ,below of=MID11] {};
    \draw [-] (system2) -- node[name=cici] {} (MID22);
    \draw [-] (system1) -- node[name=coco] {} (MID12);
    
    \draw [-] (MID22) -- node[name=passaggio2] {$\nu(t)$} (MID213);
    \draw [->] (MID213) -- node[name=passaggio2] {} (2.6,-0.2);  
\end{tikzpicture}
\caption{The interconnection between the control system and the hybrid system $\mathcal{H}_{\sigma^F}$ defined in~\cref{def:hybridcontrolpolicy}}\label{Fig:BlockDiagram}
\end{figure}

Given a winning region $\win^F\subseteq V^F$ and a winning strategy $\sigma^F:V^F_0\to V^F_1$, we now construct a set of \emph{initial winning conditions} $X_{\text{win}}\subseteq X$ and a \emph{hybrid feedback control policy} $p:\R_+\times X\times \Ldist\to U$ (as in \cref{defn:SolutionsHybrid}) to solve \cref{prob:MainProb} (realizing the orange marked transitions in \cref{fig:overview}, i.e., going from \circled{10} to \circled{11}).

We first observe that the winning region $\win^F\subseteq V^F$ naturally translates into a set of \emph{initial winning conditions} $X_{\text{win}}$ via the labeling function $\labelfunc^+$ s.t.\
\begin{equation}\label{equ:Xwin}
    X_{\text{win}}:= \{x\in X\mid \exists v\in\win^F \text{ s.t. } \labelfuncG^F(v)\cap\AP_S^+ = \labelfunc^+(x)\}.
\end{equation}

In order to translate the winning strategy $\sigma^F:V^F_0\to V^F_1$ into a hybrid control policy $p$ we take a two-step approach. We first construct a map $\Gamma$ which uses $\sigma^F$ to translate \emph{the history of} a continuous curve $\zeta:\R_+\to X$ and a disturbance function $\ldist\colon \R_+\to 2^{\AP_O}$ into a piece-wise constant function $\nu: \R_+\to V^F_1$ of $\p{1}$ vertices of $\gamegraph^F$. The hybrid controller $p$ then translates each vertex $\nu(t)\in V^F_1$ into the feedback control policy $u_w:X\rightarrow U$ associated with its (unique) label\footnote{We slightly abuse notation by writing $\labelfuncG^F(\nu(t)) = \Cc_w$ instead of $\{\Cc_w\}$. } $\ell^F(\nu(t))=\Cc_w\in\AP_C$, which is a single control proposition by construction of $\gamegraph^F$. This control policy $u_w$ is then applied to $\Sc$ via $f$. This is illustrated in \cref{Fig:BlockDiagram} and formalized in the following definition.

\begin{definition}\label{def:hybridcontrolpolicy}
 Let $\cS=(X,U,f)$ be a control system with labelling function $\labelfunc^+$ and $\allCLF$ the set of all CLFs. Consider $\sigma^F:V^F_0\to V^F_1$ a winning strategy over the final game $\game^F$, a continuous curve $\zeta:\R_+\to X$ and a disturbance function $\ldist\colon \R_+\to 2^{\AP_O}$. Then the  map $\Gamma_{\sigma^F, \zeta,\ldist}$ defines a piecewise constant  function $\nu:\R_+\to V_1^F$ such that:
 \begin{enumerate}[leftmargin=*]
  \item $\nu(0)=\sigma^F(v_0)$, where $v_0\in V^F_0$ s.t.\ $\labelfuncG^F(v_0)= \labelfunc^+(\zeta(0)) \cup \ldist(0)$, 
  \item for any discontinuity point $\tau\in \R_+$ of $\labelfunc^+(\zeta(\cdot)) \cup \ldist(\cdot)$, it holds that $\nu(\tau):=\strat^F(v)$ s.t.\ $(\nu(\tau^-),v)\in E^F$ and $\ell^F(v)= \labelfunc^+(\zeta(\tau)) \cup \ldist(\tau)$, (where $\nu(\tau^-):=\lim_{s\nearrow \tau}\nu(s)$), and 
  \item the set of discontinuity points of $\nu(\cdot)$ is contained in the set of discontinuity points of $\labelfunc^+(\zeta(\cdot)) \cup \ldist(\cdot)$.
  \end{enumerate} 
 \end{definition}
 
 Intuitively, \cref{def:hybridcontrolpolicy} models the fact that the logical layer of the hybrid controller (modelled by the game) might actuate a change in the low-level feedback control policy only when the context changes. This context change can either be induced externally (when $\ldist$ has a discontinuity point, i.e., the observation proposition changes) or when $L^+(\zeta(t))$ changes, i.e., the underlying system dynamics causing state propositions to change. Both is detected by a discontinuity point in  $\labelfunc^+(\zeta(t)) \cup \ldist(t)$.
% defined in \cref{subsec:LTL}, 
% which by definition, is a \emph{triggering point} of the $\tau$-sequence defined in item~1 of  \cref{def:hybridcontrolpolicy}.
At these triggering points (and only then), the map $\Gamma_{\sigma^F}$ mimics the move of the winning strategy $\sigma^F$ by moving to the environment vertex $v$ selected by $\sigma^F$ in $\game^F$ while respecting the current context.

 We emphasize that the definition of the map $\Gamma_{\sigma^F, \zeta,\ldist}$ is actually causal. It only uses the information from the past of $\zeta$ and $\ldist$ up to time point $t^-$ to compute $\nu(t)$. This implies that we can actually use it online to dynamically generate the signal $\nu$ from \emph{the past observations} of a state trajectory $\xi$ and the past logical disturbances $\Upsilon$, as depicted in \cref{Fig:BlockDiagram}.
 As, in this context, the state trajectory $\xi$ is not known a priory, we slightly abuse notation and refer to $\Gamma_{\sigma^F, \zeta,\ldist}$ as $\Gamma_{\sigma^F, x,\ldist}$, where $x$ is the starting point of $\xi$. %, which is the relevant information which ensures that $\Gamma$ is well defined via \cref{def:hybridcontrolpolicy}.

With this slight notation overload, we can define the \emph{final closed loop system} as follows.
 
\begin{definition}\label{def:hybridclosedloop}
Given the premises of \cref{def:hybridcontrolpolicy},
the \emph{final closed loop system} is given by
\begin{equation}\label{eq:FinalClosedLoop}
\dot x(t)=f(x(t),p(x(t),\nu(t))\,),
\end{equation}
where $p(x(t),\nu(t)):=u_w(x)\in U$ and $\nu(t)$ is dynamically generated via $\Gamma_{\sigma^F,x,\ldist}$ by interpreting (the past of) a solution $\sol_{x,p,\ldist}:\R_+\to X$ of~\eqref{eq:FinalClosedLoop} under $p$ and $\ldist$, with starting point $x\in X$, as (the past of) $\zeta$ in  \cref{def:hybridcontrolpolicy}.
\end{definition}

% 
% With the hybrid control policy defined via \cref{def:hybridcontrolpolicy}, the \emph{final closed loop system} is given by
% \begin{equation}\label{eq:FinalClosedLoop}
% \dot x(t)=f(x(t),p(x(t),\Gamma_{\sigma^F,x,\ldist}(t))\,)\,).
% \end{equation}
% We denote by $\sol_{x,p,\ldist}:\R_+\to X$ a solution of~\eqref{eq:FinalClosedLoop} under $p$ and $\ldist$, starting at $x\in X$ . 

This leads to the main result of this section establishing the correctness of our synthesis procedure.
\begin{theorem}\label{Prop:FinalProposition}
Consider a control system $\cS=(X,U,f)$ with labelling function $\labelfunc$, an LTL specification $\spec$ over the predicates $\AP_S\cup\AP_O$. Consider the final game $\game^F$,  $\allCLF$ the set of all CLFs, $\labelfunc^+$ the extended labelling function, the winning region $\win^F$ and winning strategy $\sigma^F:V^F_0\to V^F_1$. Then $x\in X_{\text{win}}$ as in \eqref{equ:Xwin} and $p$ as in \cref{def:hybridclosedloop} solve \cref{prob:MainProb}.
\end{theorem}

The proof of \cref{Prop:FinalProposition} combines all correctness results established in \cref{sec:algo:high}-\cref{subsec:solveaugmentedgame}. 
\begin{proof}
Since the plays ending in $\p{0}$ dead-ends are not winning in a game and $\strat^F$ is a winning strategy in $\game^F$, no $\strat^F$-play ends in a $\p{0}$ dead-end. 
Then, by \cref{lemma:mergedgame} and \cref{lemma:controlgraphtotal}, all possible changes in $\labelfunc^+$ (triggered by applying control policies associated with $\allCLF$) and $\ldist$ are captured by the game graph $\gamegraph^C$.
In particular, every solution $\sol_{x,p,\ldist}$ corresponds to a play $\play = v_0v_1\cdots$ in $\game^F$ such that every change in $\labelfunc^+$ and $\ldist$ corresponds to a move by $\p{1}$ to a vertex with corresponding label in $\play$.
Furthermore, as $x\in X_{\text{win}}(\win^F)$ we have $v_0\in \win^F$.
Moreover, by \cref{def:hybridcontrolpolicy}, $\play$ is a $\strat^F$-play starting from the winning region $\win^F$ of game $\game^F$.
So, $\play$ is a winning play, and hence, it always stays in $\win^F$.
This implies, $\sol_{x,p,\ldist}(t)$ also belongs to $X_{\text{win}}(\win^F)$ for all $t\in\R_+$.

By the discussed correspondence between $\sol_{x,p,\ldist}$ and play $\play$, a trace $\trace$ generated by $\sol_{x,p,\ldist}$ under $\labelfunc$ is also the trace generated by the play $\play$.
Furthermore, every play in $\game^F$ corresponds to a play in the control graph $\gamegraph^C$ as in \cref{def:controlgraph}.
Moreover, by \cref{prop:result:high-to-low}, $\trace$ satisfies~\eqref{eq:feas_control}.
Then by \cref{lemma:augmentedgamegraph} and \cref{def:augmentedGame}, $\trace$ is generated by a play in $\gamegraph^F$ satisfying $\assumpPers(\perslivegroups^F)$.
Hence, $\play$ satisfies $\assumpPers(\perslivegroups^F)$.
Moreover, as $\play$ is a winning play in $\game^F$, by \cref{prop:low-high-winnnig}, trace $\trace$ satisfies the specification $\spec$.
\end{proof}

	\section{SYNTHESIS DETAILS: HIGH-LAYER}\label{subsec:DiscreteLayer}
	The previous section described our synthesis framework and established its ability to solve \cref{prob:MainProb} in \cref{Prop:FinalProposition}.
   	The main hypotheses in this statement are  the existence of
	\begin{enumerate}[leftmargin=*]
		\item a winning strategy for the final game $\game^F$, and
		\item  a CLF $w$ for each $\cRWA$.
	\end{enumerate}
	Within this section we give a novel algorithm to efficiently solving \emph{augmented parity games} constructed in \cref{sec:algo:low-to-high}, thus tackling the first point.
	The second hypothesis is treated in subsequent~\cref{sec:ImplementDetails}, which presents the construction of \emph{feedback control policies} implementing cRWA via CLFs used in \cref{sec:algo:high-to-low}, together with the proof of the well-posedness of the arising closed loop~\eqref{eq:FinalClosedLoop}.
	%!TEX root = ../main.tex

\subsection{Augmented Reachability Games}
While an \emph{augmented parity game} can be reduced to a Rabin game (by transforming each persistent group-liveness constraints into an additional Rabin pair) and then solving the resulting Rabin game using classical algorithms~\cite{rabinGames}, this method is computationally not tractable. This is due to the fact that existing algorithms are known to become intractable very quickly if the number of Rabin pairs grows. Therefore, we leverage the recent insight that local liveness constraints on the environment player typically fall into a class of synthesis problems that allow for an efficient direct synthesis procedure~\cite{progress_groups,banerjee2023fast}. The augmented games we consider are similar to the ones discussed by Sun et al.~\cite{progress_groups}. We, however, provide a novel algorithm that tackles the full class of parity games and thereby subsumes the restricted problem class considered in \cite{progress_groups}.

The practically most efficient known algorithm to solve classical (non-augmented) parity games is Zielonka's algorithm \cite{Zielonka98}. This algorithm recursively solves reachability games for both players to compute a winning region and a winning strategy of the controller player in the original parity game. In order to mimic Zielonka's algorithm for augmented games, we first discuss an algorithm to solve augmented \emph{reachability games}. From this, our new algorithm essentially follows as a corollary.

% 
% However, the general class of games that arise in this paper have been, to the best of our knowledge, not tackled to date. Hence, we provide a new algorithm to solve such augmented parity games.
% 
An \emph{augmented reachability game} is a tuple $\game=(\gamegraph, \spec, \perslivegroups)$ where the specification $\spec = \finally T$ is to finally reach a set $T\subseteq V$ of target vertices. The new recursive algorithm that solves an augmented reachability game $\game$ is given in \cref{alg:reach}. The main idea of the algorithm is to first compute the set of vertices $A$ from which $\p{0}$ can reach $T$ even without the help of any persistent live-group constraints (line~\ref{alg:reach:Attro}) along with the corresponding strategy $\sigma$ for $\p{0}$ (line~\ref{alg:reach:stratA}). Afterwards, the algorithm computes the set of states $B$ from which $\p{0}$ has a strategy (i.e.\ $\sigma_B$) to reach $A$ with the help of a persistent live-group (lines~\ref{alg:reach:if}-\ref{alg:reach:computeB}). If this set $B$ enlarges the winning state set $A$ (line~\ref{alg:reach:B}), we use recursion to solve another augmented reachability game with target $T:=A\cup B$ (line~\ref{alg:reach:recursion}).

\begin{algorithm}[t]
    \caption{$ \reach(\gamegraph, T, \perslivegroups) $}\label{alg:reach}
    \begin{algorithmic}[1]
        \Require An augmented game $\game = (\gamegraph,\spec,\perslivegroups)$ with $\spec = \finally T$
        \Ensure Winning region and winning strategy in the augmented game $\game$  
        \State Initialize a random $\p{0}$ strategy $\strat$
        \State $A,\strat_A \gets \Attro{\gamegraph}{T}$\label{alg:reach:Attro}
        \State $\strat(v) \gets \strat_A(v)$ for every $v\in A\setminus T$\label{alg:reach:stratA}
        \For{$(\perssource,\persedges,\perstarget) \in \perslivegroups$} 
            \If{$(\perssource\setminus A) \cap \pre(A) \neq \emptyset $} \label{alg:reach:if}
                \State $B,\strat_B \gets \solve(\gamegraph|_{\persedges}, \spec_B)$\\\hspace*{6em} with $\spec_B = \finally A \vee \globally (\perssource\setminus\perstarget)$\label{alg:reach:computeB}
                \If{$B \not\subseteq A$}\label{alg:reach:B}
                    \State $\strat(v) \gets \strat_B(v)$ for every $v\in B\setminus A$\label{alg:reach:stratB}
                    \State $C,\strat_C\gets \reach(\gamegraph, A\cup B, \perslivegroups)$ 
                    \State $\strat(v) \gets \strat_C(v)$ for every $v\in C\setminus (A\cup B)$
                    \State \Return $(C,\strat)$\label{alg:reach:recursion}
                \EndIf
            \EndIf
        \EndFor
        \State \Return $A,\strat$ \label{alg:reach:end}
    \end{algorithmic}
\end{algorithm}

Within \cref{alg:reach}, we use the following notation. Given a game graph $\gamegraph = (\vertices,\edges)$ and a persistent live-group $(\perssource,\persedges,\perstarget)$, we write $\gamegraph|_{\persedges}$ to denote the restricted game graph $(\vertices,\edges')$ such that $\edges'\subseteq\edges$ and for every edge $e = (\altvertex,\vertex)\in \edges'$, either $e\in\persedges$ or there is no edge in $\persedges$ starting from $\altvertex$.
Furthermore, $\pre(T)\subseteq V$ is the set of vertices from which there is an edge to $T$.

% We also use some standard methods to solve simple games.
For a set $T$ of vertices, the attractor function $\Attri{\gamegraph}{T}$ solves the (non-augmented) reachability game $(G,\finally T)$. I.e., it returns the attractor set $A:=\attri{\gamegraph}{T}\subseteq V$ and a  attractor strategy $\strat_A$ of $\p{i}$. Intuitively, $A$ collects all vertices from which $\p{i}$ has a strategy (i.e., $\strat_A$) to force every play starting in $A$ to visit $T$ in a finite number of steps.
Moreover, the function $\textsc{Solve}(\gamegraph,\spec)$ returns the winning region and a winning strategy in a game $(\gamegraph,\spec)$ with $\spec = \finally A \vee \globally \neg T$ for some $A,T\subseteq V$.  Both the functions $\textsc{ATTR}$ and $\textsc{Solve}$ solve classical synthesis problems with standard algorithms (see e.g.\ \cite{paritygames}). For the sake of a complete prove we note that \textsc{Solve} can be implemented using the following remark.

\begin{remark}\label{rem:safety}
Given a game $\game = (\gamegraph = (\vertices,\edges), \spec)$ where $\spec = \finally A \vee \globally S$ for some $A,S\subseteq V$, one can reduce the game to a smaller safety game $(\gamegraph', \spec' =\globally S')$, where $S' = S\cup\{v_A\}$ and $\gamegraph'$ is the game graph obtained from $\gamegraph$ by merging all vertices in $A$ to a single new sink vertex $\vertex_A$, i.e., all incoming edges to $A$ are retained but $\vertex_A$ has only one outgoing edge that is $(\vertex_A,\vertex_A)$. In such a game, the winning region is $\vertices\setminus\attrl{\gamegraph'}{V\setminus S'}$, see~\cite{paritygames}.
\end{remark}

% The main idea of the algorithm is to first compute the set of vertices from which $\p{0}$ can reach $T$ even without the persistent live-group constraints; and then find a persistent live-group using which $\p{0}$ can make a play from other vertices to reach this set. Finally, use recursion to solve another augmented game with objective to reach the current winning vertices, i.e, the vertices from which $\p{0}$ can make the play to visit $T$ using this persistent live-group constraint. 
% The correctness of the algorithm is stated in next theorem.
% 

With this, we can prove the correctness of \cref{alg:reach}.
 
\begin{restatable}{theorem}{revertexbuechi}\label{thm:persGame}
Given an augmented game  $\game = (\gamegraph,\spec,\perslivegroups)$ with $\spec = \finally T$, the algorithm $\reach(\gamegraph, T, \perslivegroups)$ returns the winning region and a winning strategy in game $\game$.
Moreover, the algorithm terminates in $\mathcal{O}(\abs{\perslivegroups}\cdot\abs{\vertices}\cdot\abs{\edges})$ time.
\end{restatable}
\begin{proof}
    Suppose $\win$ be the winning region in the augmented game $\game$. 
    Using induction on the number of times $\reach(\cdot)$ is called, we show that the set returned by the algorithm is indeed $\win$, and the updated  strategy $\strat$ returned by the algorithm is a winning strategy in $\game$. 
    
    \paragraph*{Base case:}
    If $\reach(\cdot)$ is never called, i.e., the algorithm returned $(A,\strat)$ in line~\ref{alg:reach:end}. Hence, we need to show that $A=\win$.
    
    First, let us show that $A \subseteq \win$.
    By the definition of attractor function $\Attro{\gamegraph}{T}$, every $\strat_A$-play from $A$ eventually visits $T$, and hence, satisfies $\spec$ (which is stronger than $\assumpPers(\perslivegroups) \Rightarrow \spec$). Therefore, every vertex in $A$ is trivially winning in $\game$, and hence, $A \subseteq \win$.
    
    Now, for the other direction, suppose $\vertex$ be a vertex such that $\vertex\not \in A$. It is enough to show that $\vertex\not \in \win$.
    As $\vertex\not \in A = \attro{\gamegraph}{T}$, $\p{0}$ can not force the plays to visit~$T$. 
    If $q\not\in \perssource$ for every $(\perssource,\persedges,\perstarget)\in \perslivegroups$, then the persistent group-liveness constraints are not relevant for vertex $\vertex$. 
    % Hence, $\vertex\not\in\win$. 
    Now, suppose $\vertex\in\perssource$ for some $(\perssource,\persedges,\perstarget)\in \perslivegroups$.
    As the algorithm did not reach line~\ref{alg:reach:recursion}, for every persistent live-group, one of the conditional statements, the one in line~\ref{alg:reach:if} or the one in line~\ref{alg:reach:B}, is not satisfied.
    If the statement in line~\ref{alg:reach:if} is not satisfied, i.e., $(\perssource\setminus A) \cap \pre(A) = \emptyset $, then there is no edge from $\perssource\setminus A$ to $A$, and hence, this persistent live-group constraint does not help in reaching $A$ from $\vertices\setminus A$ anyway. 
    
    Next, if the statement in line~\ref{alg:reach:if} is not satisfied, then it holds that $B \subseteq A$. Hence, $\vertex\not \in B$. 
    As $B$ is the winning region for game $(\gamegraph|_{\persedges}, \spec_B)$ and such a game is determined~\cite{paritygames}, $\p{1}$ has a strategy $\strat_1$ such that every $\strat_1$-play in this game starting from $\vertex$ satisfies $\neg \spec_B = \globally \neg A \wedge \finally (\perstarget\cup \vertices\setminus\perssource)$. Therefore, every $\strat_1$-play trivially satisfies $\assumpPers(\perssource,\persedges,\perstarget)$ without ever reaching~$A$. Hence, if $\p{1}$ sticks to strategy $\strat_1$, $\p{0}$ can not make the plays from $\vertex$ visit $A \supseteq T$ using this constraint. 
    Therefore, in any case, $\p{0}$ has no strategy that can enforce a play from $\vertex$ to satisfy $\assumpPers(\perslivegroups) \Rightarrow \finally T$. Hence, $\vertex\not \in \win$.

    Now, let us show that the returned strategy $\strat$ is indeed a winning strategy in $\game$. 
    As $\strat_A$ is the attractor strategy to reach $T$, line~\ref{alg:reach:stratA}, it is easy to verify that every $\strat$-play starting from $A\setminus T$ eventually visits $T$, and hence satisfies $\spec$.
    Therefore, every $\strat$-play from $A$ is winning.

    \paragraph*{Induction case:}
    Suppose the algorithm returned $(C,\strat)$ in line~\ref{alg:reach:recursion} for some $(\perssource,\persedges,\perstarget)\in \perslivegroups$. By induction hypothesis, $C$ is the winning region and $\strat_C$ is a  winning strategy in the augmented game $\game_C = (\gamegraph,\spec_C, \perslivegroups)$ with $\spec_C = \finally  (A\cup B)$. 
    
    First, let us show that $\win \subseteq C$.
    By the definition of attractor set $\attro{\gamegraph}{\cdot}$, it is easy to see that $T\subseteq A$. So, every play in $\gamegraph$ satisfies $\finally T \Rightarrow \finally (A\cup B)$. Therefore, a winning play in augmented game $(\gamegraph, T, \perslivegroups)$ is also winning in augmented game $(\gamegraph, A\cup B, \perslivegroups)$. Therefore, $\win \subseteq C$.

    Now, for the other direction, let us first show that $B \subseteq \win$.
    As $\strat_B$ is a winning strategy in game $\game_B$, every $\strat_B$-play $\play$ starting in $B$ satisfies $\spec_B$. 
    By definition of $\spec_B$, either $\play$ satisfies $\finally A$ or it satisfies $\globally (\perssource\setminus\perstarget)$.
    Furthermore, as $\play$ is a play in $\gamegraph|_{\persedges}$, it satisfies 
    $\globally (\perssource\wedge \assumpC(\persedges))$.
    Hence, if $\play$ satisfies $\assumpPers(\perssource,\persedges,\perstarget)$, then it also satisfies $\finally \perstarget$.
    Therefore, $\play$ can not satisfy both $\assumpPers(\perssource,\persedges,\perstarget)$ and  $\globally (\perssource\setminus\perstarget)$. As a consequence, $\play$ satisfies $\assumpPers(\perssource,\persedges,\perstarget) \Rightarrow \finally A$.
    Furthermore, as we know, $A\subseteq \win$. 
    Therefore, $\play$ satisfies $\finally A \Rightarrow \finally \win$, and hence, satisfies $\assumpPers(\perssource,\persedges,\perstarget) \Rightarrow \finally \win$. So, every $\strat_B$-play starting in $B$ satisfies $\assumpPers(\perslivegroups) \Rightarrow \finally \win$. 
    Then, one can construct a $\p{0}$ strategy $\strat_0$ (i.e., the one that uses $\strat_B$ until the play reaches the winning region~$\win$ of game $\game$, and then switches to a winning strategy of game $\game$) such that every $\strat_0$-play starting in $B$ satisfies the following
    \[(\assumpPers(\perslivegroups)\Rightarrow \finally \win) \wedge \globally (\win \wedge \assumpPers(\perslivegroups)\Rightarrow \finally T),\]
    and hence, satisfies $\assumpPers(\perslivegroups) \Rightarrow \finally T$. Therefore, $B \subseteq \win$.

    Now, let us the other direction for induction case, i.e., $C\subseteq \win$. 
    As $B\subseteq \win$ and $A\subseteq \win$ as proven by the arguments given in base case, 
    it holds that $A\cup B \subseteq \win$.
    So, every play in $\gamegraph$ satisfies $\finally (A\cup B) \Rightarrow\finally \win$.
    Furthermore, as $\strat_C$ is a winning strategy in game $\game_C$, every $\strat_C$-play starting in $C$ satisfies $\assumpPers(\perslivegroups)\Rightarrow \finally (A\cup B)$, and hence, satisfies $\assumpPers(\perslivegroups)\Rightarrow \finally \win$.
    Then, as in the last paragraph, one can construct a $\p{0}$ strategy $\strat_0$ (i.e., the one that uses $\strat_C$ until the play reaches the winning region~$\win$ of game $\game$, and then switches to a winning strategy of game $\game$) such that every $\strat_0$-play starting in $C$ satisfies the following
    \[(\assumpPers(\perslivegroups)\Rightarrow \finally \win) \wedge \globally (\win \wedge \assumpPers(\perslivegroups)\Rightarrow \finally T).\]
    Hence, every $\strat_0$-play starting in $C$ satisfies $\assumpPers(\perslivegroups)\Rightarrow \finally T$. Therefore, $C \subseteq \win$.

    Now, let us show that the returned strategy $\strat$ in \cref{alg:reach:recursion} is also a winning strategy in game $\game$.
    As $\strat$ is follows strategy $\strat_C$ for vertices in  $C\setminus (A\cup B)$, every $\strat$-play from $C\setminus (A\cup B)$ eventually visits $A\cup B$ when $\assumpPers(\perslivegroups)$ holds.
    Now, let $\strat_M$ be the updated strategy until line~\ref{alg:reach:stratB}.
    Then, from line~\ref{alg:reach:stratA},\ref{alg:reach:stratB}, it is easy to see that $\strat(v) = \strat_M(v)$ for every vertex $v$ in $A\cup B$.
    As $\strat_B$ is a winning strategy in game $\game_B$, using line~\ref{alg:reach:stratB} and the discussion above, every $\strat$-play from $B\setminus A$ eventually visits $A$ when $\assumpPers(\perslivegroups)$ holds.
    Then, using arguments of base case, every $\strat$-play from $A\setminus T$ eventually visits $T$.
    Therefore, in total, as $\strat$ is a  strategy, every $\strat$-play from $C$ eventually visits $T$ when $\assumpPers(\perslivegroups)$ holds.
    Hence, $\strat$ is indeed a winning strategy in game $\game$.
    
    \paragraph*{Time complexity:}
    Let $k$ be the number of times $\reach(\cdot)$ is called. If $T=\vertices$, then $A = \vertices$, and hence, $\perssource\setminus A = \emptyset$ for every $(\perssource,\persedges,\perstarget)\in \perslivegroups$, and hence, $\reach(\cdot)$ will never be called. Furthermore, if $T\neq \vertices$, then, by definition of $\attro{\gamegraph}{\cdot}$, it holds that $T\subseteq A$. So, in line~\ref{alg:reach:if}, we keep adding at least one vertex to the target for the next call of $\reach(\cdot)$. Hence, $k$ can be at most $\abs{\vertices}$.
    Moreover, in each iteration, we might need to solve game $(\gamegraph|_{\persedges}, \spec_B)$ for each $(\perssource,\persedges,\perstarget)\in \perslivegroups$; and using \cref{rem:safety}, solving such a game can be reduced to computing an attractor function $\attrl{\gamegraph}{\cdot}$. As computing such an attractor function takes
    $\mathcal{O}(\abs{\edges})$ time~\cite{paritygames}, the algorithm takes $\mathcal{O}(\abs{\perslivegroups}\cdot\abs{\vertices}\cdot\abs{\edges})$ time in total.
\end{proof}

\subsection{Augmented Parity Games}
Zielonka's algorithm~\cite{Zielonka98} solves classical parity games by recursively using 
attractor functions $\Attro{\gamegraph}{T}$ and $\Attrl{\gamegraph}{T}$. %; and $\Attro{\gamegraph}{T}$ in such games returns the the winning region and a winning strategy in the reachability game with specification $\finally T$. 
The only difference between the attractor function $\Attro{\gamegraph}{T}$ and our new function $\reach(\gamegraph,T,\perslivegroups)$ from \cref{alg:reach} is the utilization of augmented live groups to solve reachability games. To solve an augmented parity game $(\gamegraph, \spec, \perslivegroups)$, one can therefore simply replace every use of $\Attro{\gamegraph}{T}$ with $\reach(\gamegraph,T,\perslivegroups)$ within Zielonka's algorithm. Due to \cref{thm:persGame}, the resulting algorithm correctly solves augmented parity games and returns a  strategy, summarized in the following corollary.
% Therefore, we have the following results for augmented parity games which is similar to that of parity games.
\begin{corollary}\label{corollary:solve_augmeneted_parity}
An augmented parity game with game graph $(\vertices,\edges,\labelfuncG)$ and priority function $\priority\colon\vertices\to[0,d]$ can be solved in $\mathcal{O}\left(\abs{\perslivegroups}\cdot\abs{\vertices}^{d+\mathcal{O}(1)}\right)$ time.
\end{corollary}

% In our implementation, we use edge-labelled games (and game graphs), i.e, the labelling function $\labelfunc\colon \edges \rightarrow 2^\AP$ maps each edge to a set of propositions, instead of state-labelled ones as defined in \cref{section:games_strategy_template}.
% That is because the parity games constructed from LTL formulas (using \texttt{ltlsynt}~\cite{spot_ltlsynt}) are usually smaller and more succinct when they are edge-labelled.
% However, note that \cref{lemma:strategytemplate} and \cref{thm:persGame} use algorithms for state-labelled games.
% Hence, we use these results on edge-labelled games by converting them into state-labelled ones using standard methods, i.e, by adding a dummy vertex for each edge.

	\section{SYNTHESIS DETAILS: LOW-LEVEL}\label{sec:ImplementDetails}

	This section illustrates an efficient and flexible numerical method to design CLFs which can then be used to design feedback-control policies via \cref{lemma:CLFBasedFeedback}. We show that the arising closed-loop exhibits existence of solutions from every feasible initial point and we discuss boundedness of solutions.
	\subsection{Synthesis of Control Policies from cRWAs}\label{subsec:ContinuousLayer}
	%!TEX root = ../main.tex

% 
It is well-known that the problem of synthesizing CLFs (in the sense of \cref{sec:algo:high-to-low}) for general nonlinear control systems (as in Definition~\ref{defn:ControlSystem}) over a generic state space $X\subseteq\R^{n_x}$ solving a generic cRWA problem $\RWA=(\contextRWA,\reachRWA,\avoidRWA)$ is numerically intractable~\cite{BloTsi2000}. For this reason, particular characteristics of the system and its dynamics need to be exploited for tractability.
In this section, we therefore restrict the discussion to systems with \emph{affine dynamics}, as mature computational solutions exist for this systems class. In particular, we present a novel approach to controller synthesis for cRWA problems over affine dynamical systems, by means of semidefinite optimization, considering a class of quadratic control Lyapunov functions.

While this only gives a construction for the top-down interface in \cref{sec:algo:high-to-low} for affine dynamical systems, we note that our overall hybrid controller synthesis approach discussed in \cref{sec:ControlStrategy} and summarized in \cref{fig:overview} can be applied to any dynamical system for which the generated cRWA problem can be solved. In particular, recent optimization-based approaches for enforcing logical constraints on more general nonlinear systems (see, e.g.~\cite{Xiao2021,Jagtap2021,Clark21}) can be utilized. We leave the integration of these methods into our synthesis framework for future work.
\begin{assumption}\label{Assump:AffineSubsystems}
 The control system $\cS=(X,U,f)$ has \emph{affine  dynamics} of the form
\begin{equation}\label{eq:AffineDynamics}
f(x,u)\coloneqq Ax+Bu+g,
\end{equation} 
for some $A\in \R^{n_x\times n_x}$, $B\in \R^{n_x\times n_u}$ and $g\in \R^{n_x}$. Moreover, we suppose that the input space is a convex polytope, i.e. $U=\polyhedron(p_U,H_U):=\{x\in \R^{n_x}~:~ H_U^\top (x-p_U)\leq_c\mathbf{1}\}$, for some $h_U$ and $H_U$ of appropriate dimensions.
\end{assumption}
In addition, we restrict the shape of the state-space regions linked to state propositions $\AP_S$.
\begin{assumption}\label{Assump:PropositionShape}
Given a state proposition $\Tc\in\AP_S$ its corresponding state-space region is either ellipsoidal of the type $\ellipsoid(q,S) = \{x\in \R^{n_x}~:~(x-q)^\top S(x-q)\leq1\}$ or a convex polytope $\polyhedron(p,H) =\{x\in \R^{n_x}~:~ H^\top (x-p)\leq_c\mathbf{1}\}$, where $S\in\R^{n_x\times n_x}$ is a symmetric positive semidefinite matrix, $q,p\in\R^{n_x}$ are vectors and $H\in\R^{n_x\times m}$. 
\end{assumption}

Under these assumptions, instead of searching for control Lyapunov functions all over the set of $\scrC^1$ functions, we restrict our search to \emph{quadratic functions} of the form 
\begin{equation}\label{eq:quad_lyap_fun}
	w(x)=(x-x_c)^\top P (x-x_c),
\end{equation}
where $x_c\in X$ is the \emph{center of $w$} and $P\in\R^{n_x\times n_x}$, $P\succ 0$. 

 Inspired by the results in~\cite{He2020}, we present a method to design a CLF $w(x)$ in the form of \eqref{eq:quad_lyap_fun} \emph{associated} with a cRWA problem $\RWA=(\contextRWA,\reachRWA,\avoidRWA)$ (as in \cref{def:RWAtoPolicy}) in three steps:
\begin{enumerate}[label=(\textbf{\Alph*})]
	\item \emph{Find $x_c$} such that $\reachRWA\subset \labelfunc(x_c)$  and $\avoidRWA\cap \labelfunc(x_c) = \emptyset$.\label{Item:1Control}
	\item \emph{Find a safe set} $\safeset\subseteq X$ such that $x_c\in\safeset$ and $\avoidRWA\cap \labelfunc(x) = \emptyset$ for all $x\in\safeset$.\label{Item:2Control}
	\item \emph{Construct a CLF $w$} such that its basin of attraction is safe, i.e., $X_w\subseteq \safeset$. \label{Item:3Control}
\end{enumerate}
These steps must be performed with awareness of the context $\contextRWA$ and the changes that it causes in the continuous state space.
 First, Item~\ref{Item:1Control} is
a necessary condition for the existence of a CLF that generates a feasible controller for $\RWA$. However, given that the set difference between the convex regions where $\reachRWA$ and $\avoidRWA$ hold is potentially non-convex, checking whether such $x_c$ exists is a very difficult problem. To avoid resorting to global optimization strategies such as branch-and-bound algorithms, we introduce another assumption.
\begin{assumption}\label{assum:TargetAssump}
	Given a cRWA problem $\RWA=(\contextRWA,\reachRWA,\avoidRWA)$, for all $x\in X$ such that  $\reachRWA\subset \labelfunc(x)$  we have $x\notin \mathcal{E}_\avoidRWA$, where $\mathcal{E}_\avoidRWA\subset 2^X$ is an ellipsoidal regions associated with a proposition in $\avoidRWA$.
\end{assumption}
Assumption~\ref{assum:TargetAssump} requires that any ellipsoidal set that is to be avoided in $\RWA$ does not intersect the region associated to $\reachRWA$, i.e. the region to be reached. In prctice, if it is not the case, one can replace ellipsoidal obstacles by polytopic over-approximations.
\begin{lemma}\label{lemma:TechnicalLMI1}
A point $x_c$ satisfying Item~\ref{Item:1Control} exists if the following optimization problem is feasible:
\begin{align}
	\textrm{}&\qquad x_c\in X\subset \R^{n_x}\quad \text{s.t.}\\
	\forall~ \ellipsoid_i(q_r,S_r)\in \mathcal{E}_\reachRWA  &\qquad \begin{bmatrix}
		1 & \bullet\\
		x_c-q_r & S_r^{-1}
	\end{bmatrix}\succ0,\label{eq:find_xc:reach_e}
	\end{align}
	\begin{align}
	\forall~ \polyhedron_j(p_r,H_r)\in \mathcal{P}_\reachRWA,  &\qquad H_r^\top (x_c-p_r) < \mathbf{1},\label{eq:find_xc:reach_p}\\
	\quad \forall~ \polyhedron_k(p_a,H_a)\in \mathcal{P}_\avoidRWA.  &\qquad\|H_a^\top (x_c-p_a)\|_\infty>1,\label{eq:find_xc:avoid_p}\\
	\exists u_c\in U\subseteq\R^{n_u}&\qquad Ax_c+Bu_c+g=0, \label{eq:find_xc:equilibrium}
\end{align}
where $\mathcal{E}_\reachRWA$ and $\mathcal{P}_\reachRWA$ are respectively the set of ellipsoids and polytopes associated with propositions in $\reachRWA$ while $\mathcal{P}_\avoidRWA$ is the set of polytopic sets associated with propositions in $\avoidRWA$.
\end{lemma}
\begin{proof}%
Applying the Schur Complement Lemma \cite[p.~7]{boyd1994linear}, \eqref{eq:find_xc:reach_e} becomes exactly the definition of an ellipsoid $\ellipsoid(q_r,S_r)$. The condition~\eqref{eq:find_xc:avoid_p} ensures that $\avoidRWA\cap \labelfunc(x_c) = \emptyset$. Finally,~\eqref{eq:find_xc:equilibrium} enforces that $x_c$ is a stationary point for the system under a constant input $u_c$. This last condition can be handled directly by semidefinite programs whenever $U$ is also a polytope, i.e., $U=\polyhedron(p_U,H_U)$. 
\end{proof}

To find a safe set $\safeset$ as required in Item~\ref{Item:2Control}, we shall search for the largest ellipsoid $\ellipsoid(x_c,P_\safeset)$ centered at $x_c$ and shaped through $P_\safeset\in\R^{n_x\times n_x}$.
\begin{lemma}\label{lemma:TechnicalLMI2}
The ellipsoid $\safeset=\ellipsoid(x_c,P_\safeset)$ satisfies Item~\ref{Item:2Control} if the following semidefinite program is feasible:
\begin{align}
	\min_{P_\safeset,\beta_1,\beta_2,...}\tr( P_\safeset)&\qquad\quad \text{s.t.}\\
	\forall~ \ellipsoid_i(q_a,\!P_a)\in \mathcal{E}_\avoidRWA,&	
	\begin{bmatrix}
		P_\safeset \!+\! \beta_iP_a & \!\!\!-\!P_\safeset x_c \!\!-\! \beta_iP_aq_a\\
		\bullet &\rho_i
	\end{bmatrix}\succ0, \label{eq:find_S:avoid_e}\\
	\forall~ \polyhedron_j(p_a,H_a)\in \mathcal{P}_\avoidRWA, & ~\exists h\in\cols(H_a) \quad \alpha(h) P_\safeset\succ hh^\top, \label{eq:find_S:avoid_p}
\end{align}
where $\rho_i= x_c^\top P_\safeset x_c + \beta_iq_a^\top P_aq_a -1-\beta_i$ and $\alpha(h) = (1+h^\top (p_a-x_c))^2$ and $\cols(H_a)$ denotes the set of column vectors of $H_a$. 
\end{lemma}
\begin{proof}
Note that \eqref{eq:find_S:avoid_e} is an application of the S-procedure \cite[p.~23]{boyd1994linear}, ensuring that $x\notin\ellipsoid(q_a,P_a)$ for all $x$ such that $x\in\ellipsoid(x_c,P_\safeset)$. On the other hand,~\eqref{eq:find_S:avoid_p} ensures that all polytopes in $\mathcal{P}_\avoidRWA$ have at least one hyperplane on their boundaries that separates them from the safe set $\safeset$. 
Indeed, we can prove the following statement:
\\ \noindent
For given polytope $\polyhedron(p,H)$ and ellipsoid $\ellipsoid(q,S)$, if there is $h\in\cols(H)$ such that $(1+h^\top (p-q))^2S\succ hh^\top$, we have  $\polyhedron(p,H)\cap\ellipsoid(q,S)=\emptyset$.
\\ \noindent
Indeed, since $\polyhedron(p,H)$ and $\ellipsoid(q,S)$ are convex sets, the intersection $\polyhedron(p,H)\cap\ellipsoid(q,S)$ is empty if there exists one column $h\in \R^{n_x}$ of $H$ such that 
\begin{equation}\label{eq:SeparatingHyperPlane}
h^\top(x-p)>1,\;\;\;\forall \;x\in\ellipsoid(q,S).
\end{equation}
This inequality defines  a separating hyperplane between $\ellipsoid(q,S)$ and $\polyhedron(p,H)$, since $h^\top(x-p)\leq 1$ for all $x\in\polyhedron(p,H)$, by definition.
  Since $q\in \ellipsoid(q,S)$ we have $h^\top(q-p)> 1$, and we can rewrite~\eqref{eq:SeparatingHyperPlane} as $(1+h^\top(p-q))^{-1}h^\top(x-q)< 1$, for all $x\in \ellipsoid(q,S)$. Also, since $q\in\R^{n_x}$ is the center of $\ellipsoid(q,S)$, this ellipsoid is contained also in the hyperplane defined by  $(1+h^\top(p-q))^{-1}h^\top(x-q)> -1$, and thus we have $|(1+h^\top(p-q))^{-1}h^\top(x-q)|< 1$, for all $x\in \ellipsoid(q,S)$. Thus~\eqref{eq:SeparatingHyperPlane} is equivalent to
 \[
(x-q)^\top(1+h^\top (p-q))^{-2} hh^\top(x-q)<1
 \] for all $x\in\ellipsoid(q,S)$. This, by definition, holds if and only if $(1+h^\top (p-q))^2S\succ hh^\top$, concluding the proof.
\end{proof}

Finally, having the safe set $\safeset=\ellipsoid(x_c,P_\safeset)$ fully determined, we can proceed with constructing the CLF and extracting feedback control policies from them, as required by Item~\ref{Item:3Control}.  We summarize our sufficient conditions in the following statement.
\begin{lemma}\label{lem:finalfeedbackcontrol}
Suppose that the following semidefinite program, for a given decay rate $\rho>0$, is feasible:
\begin{align}
	\max_{Z,Y,\beta_1,\beta_2,...}\tr(Z)& \qquad \text{s.t.}\\
	Z\prec &P_\safeset^{-1}\label{eq:find_w:safe}\\
	A Z+ ZA^\top +& BY+Y^\top B^\top \prec -2 \rho Z\label{eq:find_w:convergence} \\
	\forall~h_U\in \cols(H_U) \quad&\begin{bmatrix}
		Z & Y^\top h_U\\
		\bullet & (1\!+\!(p_U\!-\!u_0)^\top\! h_U)^2
	\end{bmatrix}\succ 0.\label{eq:find_w:respect_U}
\end{align}
Then, defining $P=Z^{-1}$ and $K=YP$, for the CLF defined by $w(x):=(x-x_c)^\top P (x-x_c)$ and the \emph{surrogate controller} $u(x):=K(x-x_c)+u_0$ it holds  that 
\begin{enumerate}
\item $u(x)\in U$ for all $x\in X_w$,
\item  $\inp{\nabla w(x)}{f(x,u(x))}\leq -\rho w(x)$, for all $x\in X_w$.\label{Item:TechnicalFinalLemmaFeedback}
\end{enumerate}
In particular, the function $w$ satisfies conditions in Item~\ref{Item:3Control}.
\end{lemma}
\begin{proof} First,~\eqref{eq:find_w:safe} ensures safety as, inverting both sides of the inequality implies that $X_w(1)=\ellipsoid(x_c,P)\subset \safeset$. Then~\eqref{eq:find_w:convergence} ensures the descent condition \eqref{eq:Decreasing}. Condition \eqref{eq:find_w:respect_U} implies that $u(x)\in U =\polyhedron(h_U,H_U)$ for all $x\in X_w(1)$. To show that, consider a $h_U\in \cols(H_U)$ and multiplying the first line and column of the matrix in \eqref{eq:find_w:respect_U} by $P$ and apply the Schur Complement Lemma. The result is the equivalent matrix inequality $(1+h_U^\top (p_U-u_c))^2P\succ K^\top h_Uh_U^\top K$. Multiplying it to the right by $(x-x_c)$ and to the left by $(x-x_c)^\top$  while using the assumption that $x\in X_w(1)=\ellipsoid(x_c,P)$ yields $(1+h_U^\top (p_U-u_c))^2\succ (x-x_c)^\top K^\top h_Uh_U^\top K(x-x_c) $, which can also be rewritten as $|h_U^\top (K(x-x_c)-p_U+u_c)|<1$. By definition, this inequality being fulfilled for all $h_U\in\cols (H_U)$ is equivalent to $u(x)\in\polyhedron(p_U,H_U)$.
\end{proof}
Putting Lemmas~\ref{lemma:TechnicalLMI1},~\ref{lemma:TechnicalLMI2} and~\ref{lem:finalfeedbackcontrol} together, it can be seen that the controller $u(x)$ constructed in \cref{lem:finalfeedbackcontrol} is a feedback control policy satisfying \cref{lemma:CLFBasedFeedback}, and hence also \cref{prop:result:high-to-low}.

% Summarizing, in this subsection we sketched our CLF-based design technique for the continuous state-space level. This method is applied in the subsequent Section~\ref{sec:RunningExample} on the moving robot example, while some further discussion is provided in the next subsection.

	After providing all details on the synthesis of a hybrid controller solving \cref{prob:MainProb}, we now discuss two additional issues in the correctness of this controller, which are not captured by \cref{prop:low-high-winnnig}. 
	
	%!TEX root = ../main.tex

% In this section we highlight some possible issues arising from our stabilization strategy and some alternative control techniques which can replace/complement CLF-based feedback design.
% \\
% \noindent\textbf{Existence  of Solutions:}
\subsection{Existence  of Solutions}
In our statement~of Problem~\ref{prob:MainProb} and in the control technique formalized and summarized in~\cref{Prop:FinalProposition} we state that \emph{any} (trace of) solution of the closed loop system~\eqref{eq:FinalClosedLoop} satisfies the considered LTL specification. However, we did not provide a well-posedness result establishing existence of solutions for~\eqref{eq:FinalClosedLoop}, for any initial condition and any external logical perturbation.
Indeed, it is known that closed-loop feedback systems with state-dependent piecewise-defined control input may exhibit pathological behaviors, such as chattering and sliding modes~\cite{Cortes08,goebel2012hybrid,goebel2008zeno}.

 In what follows, we thus prove the existence of solutions, in the case studied in~Section~\ref{subsec:ContinuousLayer}.

\begin{proposition}\label{Prop:ExistenceOFSolution}
Consider a control system $\cS=(X,U,f)$ with labelling function $\labelfunc$, an LTL specification $\spec$ over the predicates $\AP_S\cup\AP_O$, the final game $\game^F$ and a winning strategy $\sigma^F:V^F_0\to V^F_1$. Suppose that Assumptions~\ref{Assump:AffineSubsystems},~\ref{Assump:PropositionShape} and~\ref{assum:TargetAssump} hold, and that the set of required CLFs $\allCLF$ is build following the procedure introduced in Subsection~\ref{subsec:ContinuousLayer}. For every $x\in X_{\text{win}}$, there exists a solution $\sol_{x,p,\ldist}:\R_+\to X$ to~\eqref{eq:FinalClosedLoop} starting at $x$, in the sense of~Definition~\ref{defn:SolutionsHybrid}.
\end{proposition}
\begin{proof}
First, we recall that by Assumptions~\ref{Assump:PropositionShape} and~\ref{assum:TargetAssump} and by construction, any state proposition $\AP^+_S$ is associated to a compact (ellipsoidal or polyhedral) subset of $X$.  
The closed loop~\eqref{eq:FinalClosedLoop}, under Assumption~\ref{Assump:AffineSubsystems} can be compactly rewritten as
\[
\dot x=G(t,x)=Ax+B\,p(x,\nu(t)\,) +g,
\]
with $p(x,\Cc_w)=K_w(x-x_{cw})+u_{0w}$, for all $x\in \R^n$ and all $\Cc_w\in  \AP_C$, for some $K_w,x_{cw}$ and $u_{0w}$ of appropriate dimensions, recall~\cref{lem:finalfeedbackcontrol}.
Thus, the time-varying vector field $G:\R_+\times X\to \R^{n_x}$ is discontinuous in $t$, and recalling~\cref{def:hybridcontrolpolicy}, the discontinuity points are contained in the sequence of discontinuity points of $\labelfunc^+(\sol_{x,p,\ldist}(\cdot)) \cup \ldist(\cdot)$. We have to show that this sequence has no accumulation point, thus ruling out the so-called \emph{Zeno phenomenon}, see~\cite{goebel2012hybrid}. Since $\ldist\in \Ldist$ by assumption is piecewise constant, we have to check the behavior of discontinuities of $\labelfunc^+(\sol_{x,p,\ldist}(\cdot))$, given a fixed context $\contextRWA\subseteq \AP_O$. By construction, these discontinuities can occur only if  $\sol_{x,p,\ldist}(\cdot)$ lies at the boundaries of the regions of attraction of the CLFs $w\in \allCLF$, with $w$ associated to a $\cRWA$ with context $\contextRWA$, i.e. the CLFs that can be activated at that instant of time. For the boundaries of these region of attractions, the vector field $G$ satisfies a tranversability condition 
\[
n(x)^\top G(t,x)<0,
\]
where $n(x)$ is the \emph{normal vector} to the ellipsoid $\Xc_w$ in $x$, i.e. the vector field is ``pointing inward'' the set $\Xc_w$. This follows by Item~\ref{Item:TechnicalFinalLemmaFeedback}) in Lemma~\ref{lem:finalfeedbackcontrol}. 
This fact, also called \emph{patchy vector field} property is a sufficient condition to ensure existence of solutions (in the sense of Definition~\ref{defn:SolutionsHybrid}), as proven in~\cite[Proposition 3.1]{AncBre99}, to which we refer for the details. The completeness of solutions, i.e. the fact that any solution is well-defined on the whole positive real line $\R_+$, follows by the fact that, as proven in~\cref{Prop:FinalProposition}, by \cref{def:hybridcontrolpolicy}, a winning play $\play$ always stays in $\win^F$.
This implies, $\sol_{x,p,\ldist}(t)$ also belongs to $X_{\text{win}}(\win^F)$ for all $t\in\R_+$, concluding the proof.
\end{proof}
For a more detailed discussion regarding (properties of) solutions of discontinuous differential equations and hybrid systems, we refer to~\cite{Cortes08,goebel2012hybrid,goebel2008zeno}. 
% \noindent\textbf{Preventing Instability:}
\subsection{Preventing Instability}
As said, since the external environment can change at any instant of time, the closed loop system~\eqref{eq:FinalClosedLoop} exhibits \emph{hybrid} behavior.  This may lead to undesired phenomena on infinite horizons, as we highlight in the following simple example.

\begin{example}\label{exp:switchedsystem}
Consider a control system of the form $\cS:=(\R^{n_x},U,f)$, and two compact target sets $\cT_1,\cT_2\subset \R^{n_x}$ such that $\cT_1\cap \cT_2=\emptyset$, and consider $\AP_S=\{\cT_1,\cT_2\}$.
We consider the following desired \emph{mode-target game specification} (for an overview on mode-target games, see~\cite{balkan2017mode}):
\begin{equation}
	\varphi:= \textstyle 
	\left(\Diamond\Box \Mc_1\implies %\bigvee_{j=1}^{N_{ti}} 
	\Diamond\Box \Tc_{1}\right)\;\wedge\;\left(\Diamond\Box \Mc_2\implies %\bigvee_{j=1}^{N_{ti}} 
	\Diamond\Box \Tc_{2}\right)  \label{eq:mtg_specExample} 
\end{equation}
where $\Mc_1,\Mc_2\in  \AP_O$ are the input atomic propositions representing the \emph{modes} activated by the external environment. Suppose to have \emph{global} CLFs $w_1,w_2:\R^{n_x}\to \R$ with respect to the target $\cT_1,\cT_2$, in the sense of Definition~\ref{def:clf}, and consider continuous $u_i:\R^{n_x}\to \R^{n_u}$ satisfying~\eqref{eq:u_CLF} globally in $\R^{n_x}\setminus X_w(c)$, for any $i\in \{1,2\}$. This provides a winning strategy for the game arising from~\eqref{eq:mtg_specExample}: we activate the feedback law $u_i$ when the mode $\Mc_i$ is active. 
% On the other hand, undesired instability phenomena could arise.
Now consider the disturbance function $\ldist:\R_+\to \AP_O$ modeling the environment behavior. Then the resulting hybrid closed-loop system can be written as
\begin{equation}\label{eq:SwitchedSystem}
\dot x(t)=g(x(t), \ldist(t))
\end{equation}
where $g(x, \Mc_i):=f(x,u_i(x))$ for $i\in \{1,2\}$. Systems of the form~\eqref{eq:SwitchedSystem} are known as \emph{switched systems}, and have been intensively studied in recent years (see~\cite{Liberzon03,goebel2012hybrid} for an overview). It is well-known that, even if the  targets $\cT_1$, $\cT_2$ are asymptotically stable for the corresponding subsystems,  the external disturbance $\ldist:\R_+\to \AP_O$ can produce unbounded solutions for some initial condition $x\in \R^{n_x}$, which is undesired in many contexts, see for example~\cite[Chapter~1]{Liberzon03}.
\end{example}

There are many possible approaches to overcome the instability problem discussed in \cref{exp:switchedsystem}. Here, we informally highlight two of them.

% \begin{itemize}[leftmargin=*]
% 	\item 
	First, consider a control system $\cS=(X,U,f)$ and an LTL specification $\spec$ over $\AP_S\cup\AP_O$. Suppose that the problem is global i.e., $X=\R^{n_x}$. Consider a large enough compact set $\Cc\subset\R^{n_x}$ such that $\Xc\subset\text{int}(\Cc)$ for all $\Xc\in \AP_S$. Consider its boundary $\partial \Cc$, add $\partial \Cc\in \AP_S$ (intuitively, a large enough ``wall''), and consider a ``new'' specification $\spec'$ defined by
	$
\spec'=\spec\,\wedge\,\globally\neg \partial \Cc$.
Thus, paying the price of considering a more ``convoluted'' specification, we force, on the logical level, the solutions of $\cS$ to stay in the compact set $\Cc$.

Second, suppose that the environment, while being unpredictable, does satisfy some assumptions on the frequency of its decisions. 
More formally, suppose there exists a \emph{dwell-time} $\tau>0$,
such that, if $t\in \R_+$ is a discontinuity point of the disturbance function $\ldist$ (i.e. an instant at which the external environment changes), we suppose that $\ldist(s)=\ldist(t)$, $\forall s\in [t,t+\tau)$.  It is well-known that, if all the subsystems are asymptotically stable, a large enough dwell-time will ensure boundedness of solution of the switched system~\eqref{eq:SwitchedSystem}. The technical details are not reported here, we refer to~\cite[Section 3.2]{Liberzon03}.
% 	\end{itemize}

While the above-mentioned approaches can provide a simple stability guarantee to the hybrid-closed loop system arising from our design method, we point out that the formal study of stability/instability phenomena induced by LTL-based control is a largely open future research direction.

	\section{EXPERIMENTAL RESULTS}\label{sec:RunningExample}
	%!TEX root = ../main.tex
In this section, we demonstrate the proposed techniques on an example.
We consider the mode-target based example introduced in \cref{section:motivatingExample} in a 2-D space. The state space for the example is constrained to the box $[0,10]\times [0,10]$, and the three target regions $\cT_1$, $\cT_2$, and $\cT_3$ are ellipsoidal balls of radius $0.2$ located at co-ordinates $(3,4)$, $(3,6)$, and $(5,5)$, respectively. The sliding door is a vertical line from $(4,0)$ to $(4,10)$. The considered dynamical model for the motion of the robot is of the form introduced in Assumption~\ref{Assump:AffineSubsystems}, with a $2$-dimensional input space.

We used our proposed techniques to solve \cref{prob:MainProb} for this example. 
All  computations  were  done  on  a  MacBook  Pro 2.5GHz with 16GB RAM.
We started by constructing the initial game $\game^I$ from specification $\spec$, as given in \cref{ex:strategyTemplate}. The initial game $\game^I$ has $51$ vertices and $182$ edges, which was constructed in $0.042$ seconds. 
Next, we computed a strategy template for the initial game, and then, we translated this strategy template into several reach-while-avoid problems which took $0.007$ seconds.
Next, we constructed the control game graph $\gamegraph^C$ with $159$ vertices and $1704$ edges in $6.13$ seconds.
Next, we constructed the final augmented game $\game^F$ with $826$ vertices and $17604$ edges in $0.652$ seconds.
Finally, we solved the final game to compute a winning strategy in $112.495$ seconds which is used as a hybrid controller in the state space.
In total, our algorithm took $120$ seconds to solve \cref{prob:MainProb} for this example.

\begin{figure*}
    \centering
    \def\svgwidth{0.8\linewidth}
    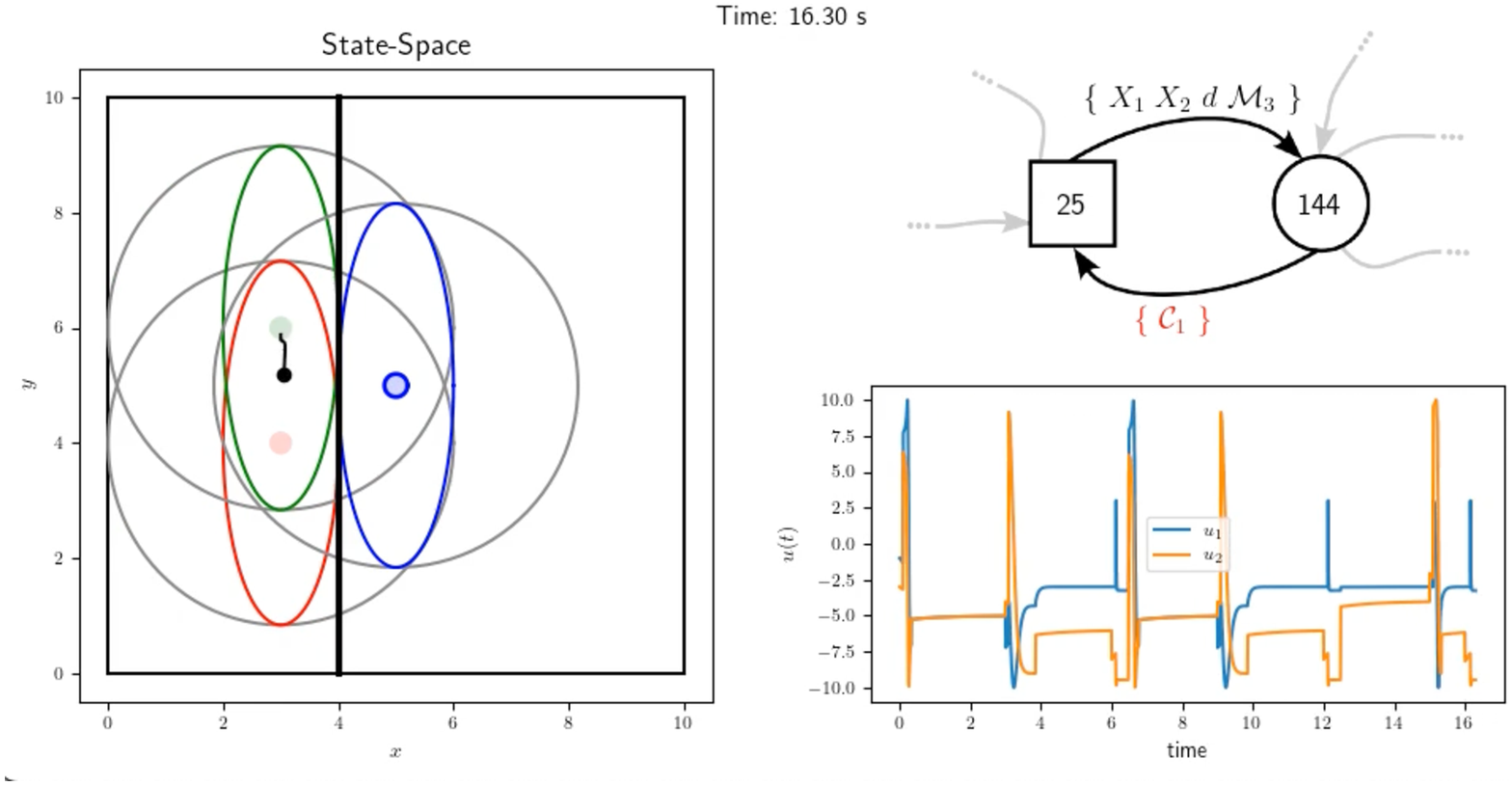
    \caption{A screenshot from the simulation video}\label{fig:simulation}
\end{figure*}

Furthermore, we also conducted a simulation\footnote{Link: \url{https://cloud.mpi-sws.org/index.php/s/Yrf2dDzspTkYm88}} of this example that uses the hybrid controller computed by our algorithm.
A screenshot from the simulation video at 16.30s is shown in \cref{fig:simulation}.
The left part of the figures describes the continuous state-space, where we have three targets, i.e., $\cT_1$ as an red colored dot (blurred), $\cT_2$ as a green colored dot (blurred), and $\cT_3$ as a blue colored dot, the robot as a black dot in motion, and two basins of attraction per each target represented by the ellipsoids around the target. The smaller ellipsoids, i.e., green, red, blue colored ones around $\cT_2$, $\cT_1$, $\cT_3$, respectively, are basins of attractions for the corresponding targets when the door is closed whereas the bigger gray ones are basins of attractions for the corresponding targets when the door is open.
Moreover, this left part also describes the current state of the system.
As we can see, the highlighted blue-colored target $\Tc_3$ indicates that currently mode $\Mc_3$ is active, the thick black line in the middle indicates that the door is closed, and the movement of the black dot from location of $\Tc_2$ towards $\Tc_1$ indicates that the robot is currently moving from target $\cT_2$ to $\cT_1$.
Furthermore, the upper-right part of the figure describes the current state of the play in the final augmented game.
Currently, the play in the game is looping between vertex $25$ and vertex $144$. 
The label of the edge from environment player's vertex (i.e., vertex $25$) indicates that the robot is currently inside the intersection of the basins of attraction $\Xc_1$ and $\Xc_2$, and currently the door is closed and mode $\Mc_3$ is active.
Furthermore, the label of the edge from controller player's vertex (i.e., vertex $144$) indicates that currently control policy associated with $\Cc_1$ is being applied persistently.
Intuitively, as mode $\Mc_3$ is active, the robot needs to reach target $\Tc_3$, and since the door is closed, the robot first need to visit target $\Tc_1$ in order to open the door. 
Specifically, in the video, the trajectory from 16.00s to 17.00s where the mode $\Mc_3$ remains consistently active can be described as follows: initially, at 16.00s, the robot was positioned at target $\cT_2$ with the door closed. Subsequently, the robot moves towards target $\cT_1$, as depicted in the screenshot shown in \cref{fig:simulation}. At 16.60s, the robot reaches $\cT_1$, resulting in the door opening. Following that, the robot proceeds towards target $\cT_3$ and successfully arrives at the target by 17.00s.

Returning to \cref{fig:simulation}, the lower-right part of the figure presents the time-responses of the two components of the control input, namely $u_1$ and $u_2$, which emerge from the hybrid feedback control policy defined in Subsection~\ref{subsec:hybridcontroller}.

	\section{CONCLUSION}\label{Sec:Conclu}
	%!TEX root = ../main.tex
In this paper we proposed a method to synthesize feedback controllers for continuous-time systems, in order to fulfill general LTL specifications. We presented our main algorithm, which, on the logical level, aims to rewrite the general problem in the form of an augmented parity game. In order to efficiently perform our proposed method, a new solving algorithm for augmented games is proposed.
On the continuous state-space level, the winning strategy is implemented via a control Lyapunov functions approach, which provides a natural and flexible feedback design for a large class of dynamical systems. %\new{Note that so-called state-space methods have provided in the last decades other concepts than Control Lyapunov Functions, such as for instance Control Barrier Functions (see, e.g.~\cite{Xiao2021,Jagtap2021,Clark21}). These concepts could also be leveraged in our framework, as they can be translated to logical constraints, which in turn can be implemented in the high-level controller. However we leave this for further work. Indeed, we focus here on the concept of Control Lyapunov Functions, which we critically need for ensuring convergence to the target regions.}

We believe that our work paves the way towards a new generation of symbolic controllers, where formal guarantees are still available, thanks to rigorous techniques both at the logical and dynamics levels; however with satisfactory scalability performances, because the (time- and space-) discretizations are computed endogenously, in an event-triggered philosophy. As further directions of research, we plan to extend our approach to more general logical/dynamical systems settings and to formally investigate and improve both  numerical complexity and theoretical conservatism of the proposed methods. In particular, we believe that our framework fits for an iterative, or active learning, approach, where the solution, and the bottlenecks, at the logical level may be used as information to guide the low-level design, and vice-versa.

	% \section*{ACKNOWLEDGMENT}
	
	\bibliographystyle{ieeetr}
	\bibliography{bib-file}

\end{document}